\documentclass[12pt]{article}
\pdfoutput=1

\usepackage[bulletsep]{collref}
\usepackage{amssymb,graphicx}
\usepackage[intlimits]{amsmath}
\usepackage{bbm}
\usepackage[small]{subfigure}

\usepackage{MnSymbol}

\usepackage{amsthm}

\usepackage{arydshln}

\usepackage{color}

\usepackage{pifont}
\newcommand{\cmark}{\ding{51}}%
\newcommand{\xmark}{\ding{55}}%

\usepackage[makeroom]{cancel}

\makeatletter \@addtoreset{equation}{section} \makeatother

\makeatletter
\let\old@startsection=\@startsection
\let\oldl@section=\l@section
\renewcommand{\@startsection}[6]{\old@startsection{#1}{#2}{#3}{#4}{#5}{#6\mathversion{bold}}}
\renewcommand{\l@section}[2]{\oldl@section{\mathversion{bold}#1}{#2}}
\makeatother

\makeatletter
\let\old@makecaption=\@makecaption
\def\@makecaption{\small\old@makecaption}
\makeatother

\newtheorem{defn}{Def}
\newtheorem{thm}{Th}
\newtheorem{expl}{Ex}
\newtheorem{exc}{Exercise}

\begin{document}


\begin{flushright}\footnotesize
\texttt{NORDITA 2017-137} \\
\texttt{UUITP-04/17}
\end{flushright}

\renewcommand{\thefootnote}{\fnsymbol{footnote}}
\setcounter{footnote}{0}

\begin{center}
{\Large\textbf{\mathversion{bold} Integrability in Sigma-Models}
\par}

\vspace{0.8cm}

\textrm{
K.~Zarembo\footnote{Also at ITEP, Moscow, Russia}}
\vspace{4mm}

\textit{Nordita,  Stockholm University and KTH Royal Institute of Technology,
Roslagstullsbacken 23, SE-106 91 Stockholm, Sweden}\\
\textit{Department of Physics and Astronomy, Uppsala University\\
SE-751 08 Uppsala, Sweden}\\
\vspace{0.2cm}
\texttt{zarembo@nordita.org}

\vspace{3mm}


\par\vspace{1cm}

\textbf{Abstract} \vspace{3mm}

\begin{minipage}{13cm}
This is a write-up of lectures on integrable sigma-models, which covers the following topics: (1) Homogeneous spaces, (2) Classical integrability of sigma-models in two dimensions, (3) Topological terms, (4) Background-field method and beta-function, (5)  S-matrix bootstrap in the $O(N)$ model, (6) Supersymmetric cosets and strings on $AdS_d\times X$.
\end{minipage}

\end{center}

\vspace{0.5cm}



\setcounter{page}{1}
\renewcommand{\thefootnote}{\arabic{footnote}}
\setcounter{footnote}{0}

\section{Introduction}

A sigma-model is a field theory wherein fields take values in a curved manifold $\mathcal{M}$. In other words, a field configuration is a map\footnote{The space-time manifold $\Sigma $, for the purpose of these notes, will almost exclusively be the 2d Minkowski space with $(+-)$ signature or the Euclidean $\mathbbm{R}^2$. } 
\begin{equation}
 X^M(\sigma ): ~\Sigma \rightarrow \mathcal{M}.
\end{equation}
Sigma-models arise as effective field theories in many physical contexts that range from low-energy QCD \cite{Weinberg:1966fm,Weinberg:1968de} to condensed-matter systems \cite{Haldane:1982rj,Haldane:1983ru,Tsvelik-book}. 

Another vast area of applications of sigma-models is string theory, where they govern string propagation in curved backgrounds.
The ultimate goal of these notes is to introduce string sigma-models relevant for holographic duality whose key feature is complete integrability. Quantization of integrable sigma-models will then be exemplified  by the S-matrix bootstrap in the $O(N)$ model \cite{Zamolodchikov:1978xm}. Although sigma-models that arise in the holographic duality are in many respects different, methods for their exact solution are conceptually the same. 

The canonical example of holography, the AdS/CFT correspondence, deals with string theory on $AdS_5\times S^5$, a non-linear manifold whose curvature is controlled by the 't~Hooft coupling of the dual gauge theory. The sigma-model on $AdS_5\times S^5$ \cite{Metsaev:1998it} is completely integrable \cite{Bena:2003wd}, a remarkable property not only of this particular model, but of a whole class of holographic string backgrounds, which underlies proliferation of integrability methods in the gauge-string duality. The goal of these notes is a step-by-step introduction to integrable sigma-models.

Given local coordinates $X^M$ on $\mathcal{M}$, the most general two-derivative Lagrangian of a sigma-model is 
\begin{equation}\label{Lagrangian-sigma-general}
 \mathcal{L}=\frac{1}{2}\left(\sqrt{|h|}\,h^{ab}G_{MN}(X)\partial _aX^M\partial _bX^N
 +\varepsilon ^{ab}B_{MN}(X)\partial_a X^M\partial _bX^N\right).
\end{equation}
Transformations that leave the metric and B-field  invariant translate into global symmetries of the sigma-model. Symmetries are neither necessary nor sufficient for integrability, but they allow one to build a large class of integrable models. The sigma-models arising in the holographic duality are precisely of this type. For this reason we concentrate on the cases when the target space $\mathcal{M}$  admits an action of a (simple) Lie group (or supergroup) $G$, and in the next section recollect basic facts about action  of Lie groups on manifolds, following \cite{Dubrovin-book}. We then describe classical integrability of sigma-models on such manifolds, following \cite{Faddeev:1987ph}. 

\section{Geometry}\label{Geometry:sec}

We start by collecting basic geometric facts about homogeneous spaces, closely following  \cite{Dubrovin-book}.

\begin{defn}
A map
$$
 T_g(x):~G\times \mathcal{M}\rightarrow \mathcal{M}
$$
defines left (right) action of group $G$ on manifold $\mathcal{M}$, if
$$
 T_gT_h=T_{gh}\qquad (T_gT_h=T_{hg})
$$
and $T_1={\rm id}$.
\end{defn}

The stability group -- or the little group -- of a point $x\in\mathcal{M}$ is defined as the set of all elements in $G$ that leave $x$ intact:
$$
 H_x=\left\{g\in G\,|\, T_g(x)=x\right\}.
$$

\begin{defn}
$\mathcal{M}$ is a left (right) homogeneous space of group $G$, if the action of $G$ on $\mathcal{M}$  is transitive, namely if
$$
 \forall x,y\in \mathcal{M}~~\exists g\in G:~T_g(x)=y.
$$
\end{defn}

Stability groups of any two points in a homogeneous space are isomorphic to one another. To see this we notice that any two points in $\mathcal{M}$, $x$ and $y$, are  related by the group action: $T_g(x)=y$ and $T_{g^{-1}}(y)=x$. Consider now an element $h\in H_x$ ($T_h(x)=x$), then
$$
 y=T_g(x)=T_gT_h(x)=T_gT_hT_{g^{-1}}(y)=T_{ghg^{-1}}(y).
$$
Therefore $ghg^{-1}\in H_y$ and the stability groups $H_y$ and $H_x$ are related by conjugation:
$$
 H_y=gH_xg^{-1}.
$$
In particular, they have the same multiplication table.

\begin{figure}[t]
\begin{center}
 \centerline{\includegraphics[width=8cm]{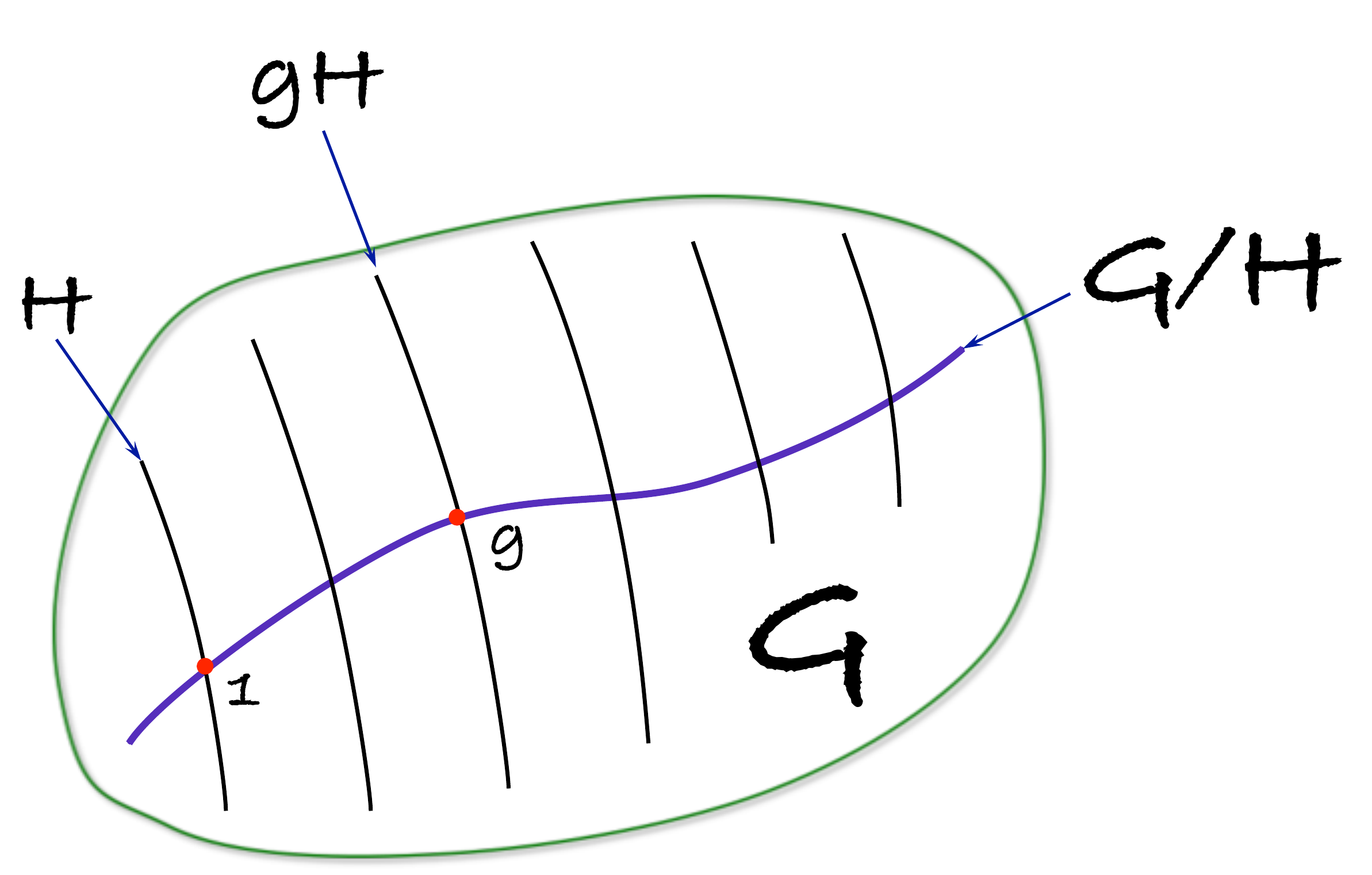}}
\caption{\label{coset}\small A group manifold foliated over equivalence classes $gH$. In the coset $G/H$, each equivalence class is projected to one point.}
\end{center}
\end{figure}

Given a  subgroup $H\subset G$, one can define a (right) coset $G/H$ as a set of equivalence classes  with respect to right multiplication by $H$:
$$
 G/H=\left\{g\sim gh\,|\,g\in G,\, h\in H\right\}.
$$
A set $gH$, obtained by multiplying  all elements of $H$ by $g$, constitutes one point in $G/H$ (fig.~\ref{coset}). One can show that for a closed Lie subgroup $H\subset G$, the coset space $G/H$ is a smooth manifold. The left coset $H\backslash G$ is defined in a similar way.

The coset $G/H$ is a left homogeneous space of $G$, with the group action defined by left multiplication:
\begin{equation}
 T_k(gH)=kgH.
\end{equation}
Interestingly, the converse is also true, in virtue of the following theorem.

\begin{thm}
 Homogeneous space $\mathcal{M}$ is isomorphic to the coset of its symmetry group $G$ by the stability group $H_{x_0}$ of any point $x_0\in\mathcal{M}$:
 $\mathcal{M}=G/H_{x_0}$.
\end{thm}
\begin{proof}
Define a map $f:\,G/H_{x_0}\rightarrow \mathcal{M}$ by
$$
 f(gH)=T_g(x_0).
$$
This map
\begin{itemize}
 \item does not depend on the representative in the equivalence class $gH$:
 $$
  T_{gh}(x_0)=T_gT_h(x_0)=T_g(x_0),
 $$
 and thus maps the whole equivalence class to a single point in $\mathcal{M}$.
 \item covers all $\mathcal{M}$. This follows from transitivity of the group action.
 \item is one-to-one. Indeed, assuming that $f(g_1H)=f(g_2H)$, we have \\ $T_{g_1}(x_0)=T_{g_2}(x_0)$ and $T_{g_2^{-1}g_1}(x_0)=x_0$. Consequently, $g_2^{-1}g_1\in H_{x_0}$. Denoting $g_2^{-1}g_1=h$ we can thus write $g_1=g_2h$ with $h\in H_{x_0}$, which proves that $g_1$,  $g_2$ belong to the same equivalence class. Different equivalence classes, therefore, map to different points on $\mathcal{M}$.
\end{itemize}

The construction does not really depend on the base point $x_0$, because the stability groups $H_x$ for different $x\in\mathcal{M}$ are all isomorphic.
\end{proof}

An algebraic characterization of homogeneous spaces as cosets of a group by its subgroup is very convenient for field-theory applications. Since the construction may look rather abstract in the beginning, we illustrate it on a number of examples.

\begin{expl}
 Group manifold itself is both left and right homogeneous space.
\end{expl}
\begin{proof}
 The group action is defined by left or right multiplication:
\begin{eqnarray}
 T_g^L(x)&=&gx
\nonumber \\
T_g^R(x)&=&xg.
\end{eqnarray}
\end{proof}

\begin{figure}[t]
\begin{center}
 \centerline{\includegraphics[width=7cm]{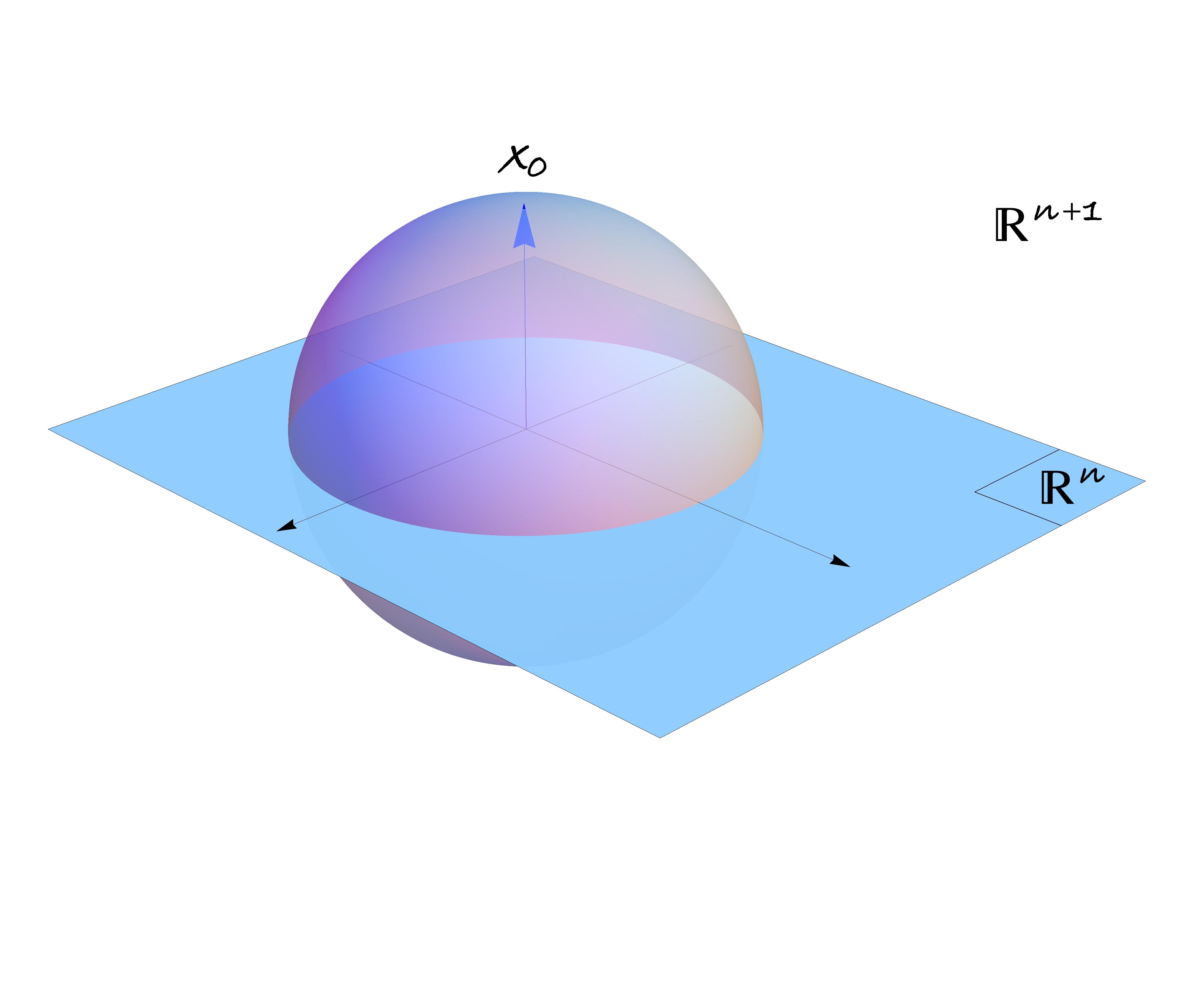}}
\caption{\label{splines}\small $n$-dimensional sphere is $SO(n+1)/SO(n)$.}
\end{center}
\end{figure}

\begin{expl}
 $n$-dimensional sphere $S^n$
\end{expl}
\begin{proof}
 The $SO(n+1)$ rotation group acts transitively on the unit sphere in $\mathbbm{R}^{n+1}$. The stability group of the north pole is the $SO(n)$ rotation group of the lateral hyperplane (fig.~\ref{splines}). Consequently,
\begin{equation}
 S^n=SO(n+1)/SO(n).
\end{equation}
\end{proof}

\begin{expl}
Complex projective space $\mathbbm{CP}^n=\left\{\mathbf{z}\sim \lambda \mathbf{z}\,|\,\lambda \in \mathbbm{C}^*,\mathbf{z}\in\mathbbm{C}^{n+1}\right\}$
\end{expl}
\begin{proof}
$SU(n+1)$ acts transitively on $\mathbbm{CP}^n$. The invariance subgroup of  the point $(1,\mathbf{0})\sim (\lambda ,\mathbf{0})$ is
$$
\left(\,
\begin{array}{ c : c cc}
 \,{\rm e}\,^{i\varphi }  & &  \\\hdashline
 & & & \\
    &        &   SU(n)\times \,{\rm e}\,^{-\frac{i\varphi }{n}}   &  \\ 
  & & & \\  
\end{array}\,\right)
\begin{pmatrix}
 1  \\ 
 0  \\ 
 \vdots  \\ 
 0  \\ 
 \end{pmatrix}
$$
Consequently, $\mathbbm{CP}^n=SU(n+1)/SU(n)\times U(1)$.
\end{proof}

\begin{expl}
 Anti-de-Sitter space
\end{expl}
\begin{proof}
 The $(d+1)$-dimensional Anti-de-Sitter space has coordinates $(x^\mu ,z)$ and the line element
\begin{equation}\label{PoincareAdS}
 ds^2=\frac{dx_\mu dx^\mu +dz^2}{z^2}\,,
\end{equation}
where  $\mu=0\ldots d-1 $ and the indices are contracted with the $-+\ldots +$ Minkowski metric (for $AdS_{d+1}$) or with Euclidean metric (for $EAdS_{d+1}$). The AdS space can be realized as a hypersurface in $\mathbbm{R}^{d,2}$ (or $\mathbbm{R}^{d+1,1}$) with embedding coordinates
\begin{eqnarray}
 X^\mu &=&\frac{x^\mu }{z}
\nonumber \\
 X^{-1}&=&\frac{x^2+z^2+1}{2z}
\nonumber \\
X^{d}&=&\frac{x^2+z^2-1}{2z}\,,
\end{eqnarray}
that satisfy
\begin{equation}\label{embeddingAdS}
 \eta _{MN}X^MX^N+1=0,
\end{equation}
where $M,N=-1\ldots d$ and $\eta _{MN}$ is the  pseudo-Euclidean metric with signature $-\mp+\ldots +$ (the upper sign is for Minkowskian AdS). 

\begin{figure}[t]
\begin{center}
 \subfigure[]{
   \includegraphics[height=5.2 cm] {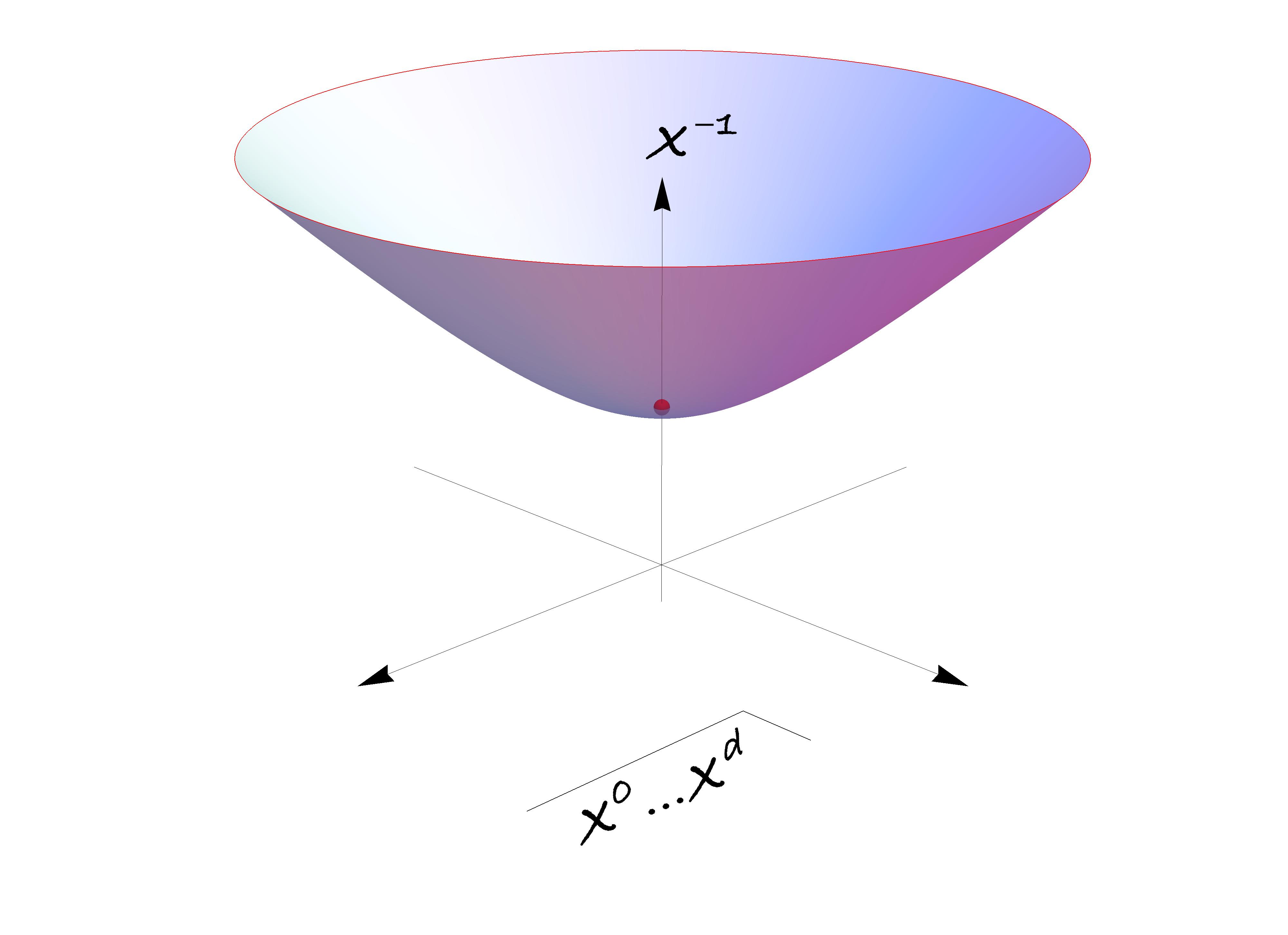}
   \label{fig2:subfig1}
 }
 \subfigure[]{
   \includegraphics[height=5.2 cm] {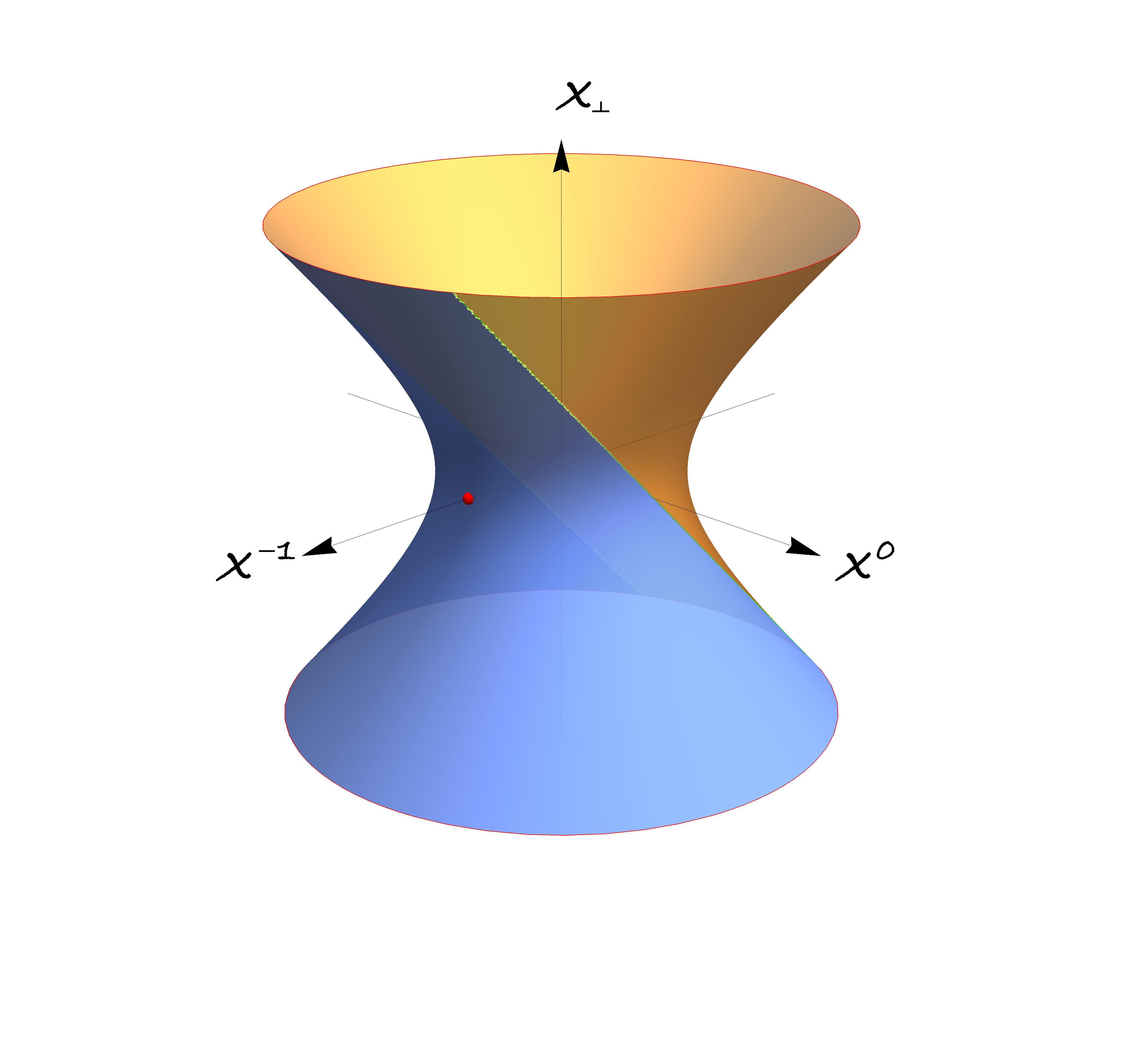}
   \label{fig2:subfig2}
 }
\caption{\label{AdSy}\small (a) $EAdS_{d+1}$ (b) $AdS_{d+1}$ }
\end{center}
\end{figure}

Euclidean Anti-de-Sitter (Lobachevsky) space $EAdS_{d+1}$ is the upper sheet of the  two-sheeted hyperboloid imbedded into the $(d+2)$-dimensional Min\-kow\-ski space, fig.~\ref{fig2:subfig1}. The hyperboloid asymptotes to the future light-cone, and the boundary of Euclidean AdS, the sphere $S^d$ at infinity shown in fig~\ref{fig2:subfig1} as a red circle, can be thought of as a set of all light rays emitted from the origin. 

Minkowski AdS, as defined by eq.~(\ref{embeddingAdS}), is a one-sheeted hyperboloid, fig.~\ref{fig2:subfig2}. The relationship between the space with the Poincar\'e metric (\ref{PoincareAdS}) and the surface (\ref{embeddingAdS}) is more intricate in this case.
The coordinates
 $(x^\mu ,z)$ cover only part of the hyperboloid, the Poincar\'e patch with 
 $X^{-1}>X^d$ (which corresponds to $z>0$). The two patches, shown in blue and yellow in fig.~\ref{fig2:subfig2}, are glued together along the horizon at $z=\infty $.  
 The lateral sections of the hyperboloid are actually time-like, and on the boundary
that has the $S^{d-1}\times S^1$ geometry, the time direction is cyclic. The space which is usually called global $AdS_{d+1}$ is
 the universal cover of the hyperboloid (\ref{embeddingAdS}) obtained by unwinding the time direction. 

The embedding (\ref{embeddingAdS})   defines a transitive action of $SO(d,2)$ on $AdS_{d+1}$ and of $SO(d+1,1)$ on $EAdS_{d+1}$. The stability group of the point shown as a red dot in fig.~\ref{AdSy} is $SO(d+1)$ for $EAdS_{d+1}$ and $SO(d,1)$ for $AdS_{d+1}$. Hence,
\begin{equation}
 EAdS_{d+1}=SO(d+1,1)/SO(d+1),\qquad 
 AdS_{d+1}=SO(d,2)/SO(d,1).
\end{equation}
\end{proof}

\section{Principal Chiral Field}  \label{PCF:sec}

The principal chiral field is a non-linear field that takes values in a group manifold:
$$
 g(x):~\Sigma \rightarrow G.
$$
One can define a Lie-algebra-valued current
\begin{equation}\label{LIcurrent}
 j_a=g^{-1}\partial _a g\in\mathfrak{g},
\end{equation}
and construct the Lagrangian as 
\begin{equation}\label{PCFLagrangian}
 \mathcal{L}=-\mathop{\mathrm{tr}}j_aj^a=-\mathop{\mathrm{tr}}j\wedge *j.
\end{equation}
Here $"\mathop{\mathrm{tr}}"$ denotes the invariant quadratic form on $\mathfrak{g}$. The overall normalization and sign are chosen so as to give conventional kinetic term to $X^A$ in the parameterization $g=\,{\rm e}\,^{iX^AT_A}$, where $T_A$ are Lie algebra generators normalized such that $\mathop{\mathrm{tr}}T_AT_B=\delta _{AB}/2$. Another commonly used normalization includes an extra factor of $1/4$.

The Lagrangian is invariant under global $G_L\times G_R$ transformations
\begin{equation}
 g(x)\rightarrow g_Lg(x),\qquad 
 g(x)\rightarrow g(x)g_R.
\end{equation}
The current is left-invariant and transforms in the adjoint under right multiplications: $j_a\rightarrow g_R^{-1}j_ag_R$. The invariance of the Lagrangian in the latter case follows from the invariance of the quadratic form, denoted by $\mathop{\mathrm{tr}}$, under the group action.

The equations of motion of the principal chiral field follow from an infinitesimal variation
$$
 \delta g=g\xi ,\qquad \xi \in\mathfrak{g}.
$$
The variation of the current is
\begin{equation}\label{varcurrent}
  \delta j_a=\partial _a\xi +[j_a,\xi ].
\end{equation}
The commutator term does not contribute to the variation of the Lagrangian, because of the invariance of the quadratic form under the adjoint action of the Lie algebra. Partial integration in the remaining derivative term yields:
$$
 \delta S=2\int_{}^{}d^2x\,\mathop{\mathrm{tr}}\xi \partial _aj^a.
$$
The equations of motion therefore are:
\begin{equation}\label{conscurrent}
 \partial _aj^a=0,
\end{equation}
and are equivalent to the conservation of current $j_a$.
The current $j_a$  is nothing but the Noether current of the right group multiplication and the equations of motion can be alternatively derived from the Noether theorem.

When written in terms of $g(x)$, the conservation of current is a second-order non-linear differential equation. Boundary conditions, for instance,
\begin{equation}\label{bcs}
 g(x^0,\pm\infty )=1,
\end{equation}
and initial conditions at some fixed time slice lead  to a well-defined Cau\-chy problem. 

However, it is far more convenient to regard the current itself as a dynamical variable. To this end, we notice that (\ref{LIcurrent}) defines a pure-gauge potential, a flat connection, whose curvature is equal to zero:
\begin{equation}\label{flatcurrent}
 \partial _aj_b-\partial _bj_a+[j_a,j_b]=0.
\end{equation}
 Unlike (\ref{conscurrent}), this equation is an identity, a consequence of  definition (\ref{LIcurrent}). Nevertheless we can treat eqs.~(\ref{conscurrent}) and (\ref{flatcurrent}) on equal footing, as two equations of motion for the two components of the  current. Once the current is known, $g(x)$ can be reconstructed from (\ref{LIcurrent}) after imposing the boundary conditions (\ref{bcs}).
 
In the index-free notations, 
\begin{eqnarray}
 &&d*j=0 \label{cons}
 \\
 &&dj+j\wedge j=0. \label{wedge}
\end{eqnarray}
These two equations can be combined into one by the following simple trick.  Let us multiply the first equation by $z$ and add to the second equation: $z(\ref{cons})+(\ref{wedge})$.  We do not loose any information by doing so, because two expressions are null as soon as their arbitrary linear combination is null.  Explicitly,
\begin{equation}\label{prelim}
 d\left(j+z*j\right)+j\wedge j=0.
\end{equation}
Here $z$ is just a dummy variable, an arbitrary complex number for instance.

The  last equation suggest to redefine the current to $j+z*j$. And indeed, using the identities
$$
 *a\wedge b=-a\wedge *b,\qquad *^2=1,
$$
valid for one-forms in two dimensions,
the last term in the equation can be expressed through the new current as well, up to a rescaling factor:
 $$
  \left(j+z*j\right)\wedge\left(j+z*j\right)=\left(1-z^2\right)j\wedge j.
 $$
A simple rescaling factor:
\begin{equation}
 L=\frac{j+z*j}{1-z^2}\,,
\end{equation}
makes the current so defined flat:
\begin{equation}
 dL+L\wedge L=0.
\end{equation}

The current $L$ is called the Lax connection. Explicitly,
\begin{equation}
 L_a=\frac{j_a+z\varepsilon _{ab}j^b}{1-z^2}\,.
\end{equation}
The equations of motion for the principal chiral field are equivalent to the condition that the Lax connection is flat \cite{Zakharov:1973pp}:
\begin{equation}\label{flatc}
 \partial _aL_b-\partial _bL_a+[L_a,L_b]=0~~\forall z.
\end{equation}
The zero curvature, or Lax representation of the equations of motion is a hallmark of integrability.

In virtue of the zero curvature condition, the  holonomy (the Wilson line) of $L$ is invariant under continuous deformations of the contour. The holonomy along a constant time slice defines the monodromy matrix:
\begin{equation}
 M(t;z)=P\exp\int_{-\infty }^{+\infty }dx\,L_1(t,x;z),
\end{equation}
which for generic complex $z$ is an element of the complexified group $G^{\mathbbm{C}}$. Translation of the contour along the time direction is a continuous deformation and therefore leaves the monodromy matrix intact.
Hence, the monodromy matrix is time-independent:
\begin{equation}
 \partial _tM(t;z)=0.
\end{equation}
A quantity that does not change with time, once the equations of motion are imposed, is called the conserved charge. The monodromy matrix defines an infinite number of those because it contains an auxiliary parameter $z$ (called the spectral parameter). The individual conserved charges are generated by expanding the monodromy matrix in $z$, for instance by the Laurent expansion at infinity:
\begin{equation}
 M(z)=\exp\left(\sum_{n=1}^{\infty }\frac{Q_n}{z^n}\right).
\end{equation}

Written explicitly in terms of currents, the monodromy matrix is
\begin{equation}
 M(z)=P\exp\int_{-\infty }^{+\infty }dx\,\,\frac{j_1-zj_0}{1-z^2}\,,
\end{equation}
and to the first order in $1/z$,
\begin{equation}
 M(z)=1+\frac{1}{z}\,\int_{-\infty }^{+\infty }dx\,j_0+\mathcal{O}\left(\frac{1}{z^2}\right).
\end{equation}
The first charge in the hierarchy, $Q_1$, therefore, is just the Noether charge of right group multiplications.

\begin{exc}
Compute the second (Yangian) charge $Q_2$ by expanding the monodromy matrix to the second order in $1/z$.
\end{exc}

\begin{exc}
Derive the Noether current of the left group multiplications. Show that the corresponding charge appears at the first order in the Taylor expansion of the monodromy matrix at $z=0$.
\end{exc}

Higher conserved charges obtained by expanding the monodromy matrix at $z=\infty $ are non-local, they cannot be represented by integrals of local densities. Local conserved quantities are generated by Taylor expansion  at $z=\pm 1$. The conventional energy and momentum of the sigma-model, however,  cannot be obtained from the monodromy matrix. Integrability of the principal chiral field is more consistent with an alternative canonical structure that does not follow from the Lagrangian (\ref{PCFLagrangian}), but whose Hamiltonian and momentum appear as first terms in the expansion of the monodromy matrix at $z=\pm 1$ \cite{Faddeev:1985qu}.

The zero-curvature representation  has far-reaching consequences. It can be used to explicitly construct solutions of the equations of motion, or to find the action-angle variables, in which the dynamics is trivial. The counterpart of the Lax representation in quantum theory is the Algebraic Bethe Ansatz \cite{Faddeev:1996iy}, which defines the building blocks of the quantum spectrum. The usual issues of UV regularization and definition of the true ground state are very noexemplifyn-trivial in integrable QFTs,  and it is often easier to reconstruct  the final answer from the constraints imposed by integrability on the dynamics, rather than to solve the model directly. Unfortunately, many steps of this bootstrap procedure are model-dependent and may also depend on the boundary conditions. We exemplify the integrability bootstrap in sec.~\ref{ONexact:sec} by the exact solution  of the $O(N)$ sigma-model. 

\section{Symmetric cosets}\label{Symmcos:sec}

\begin{figure}[t]
\begin{center}
 \centerline{\includegraphics[width=8cm]{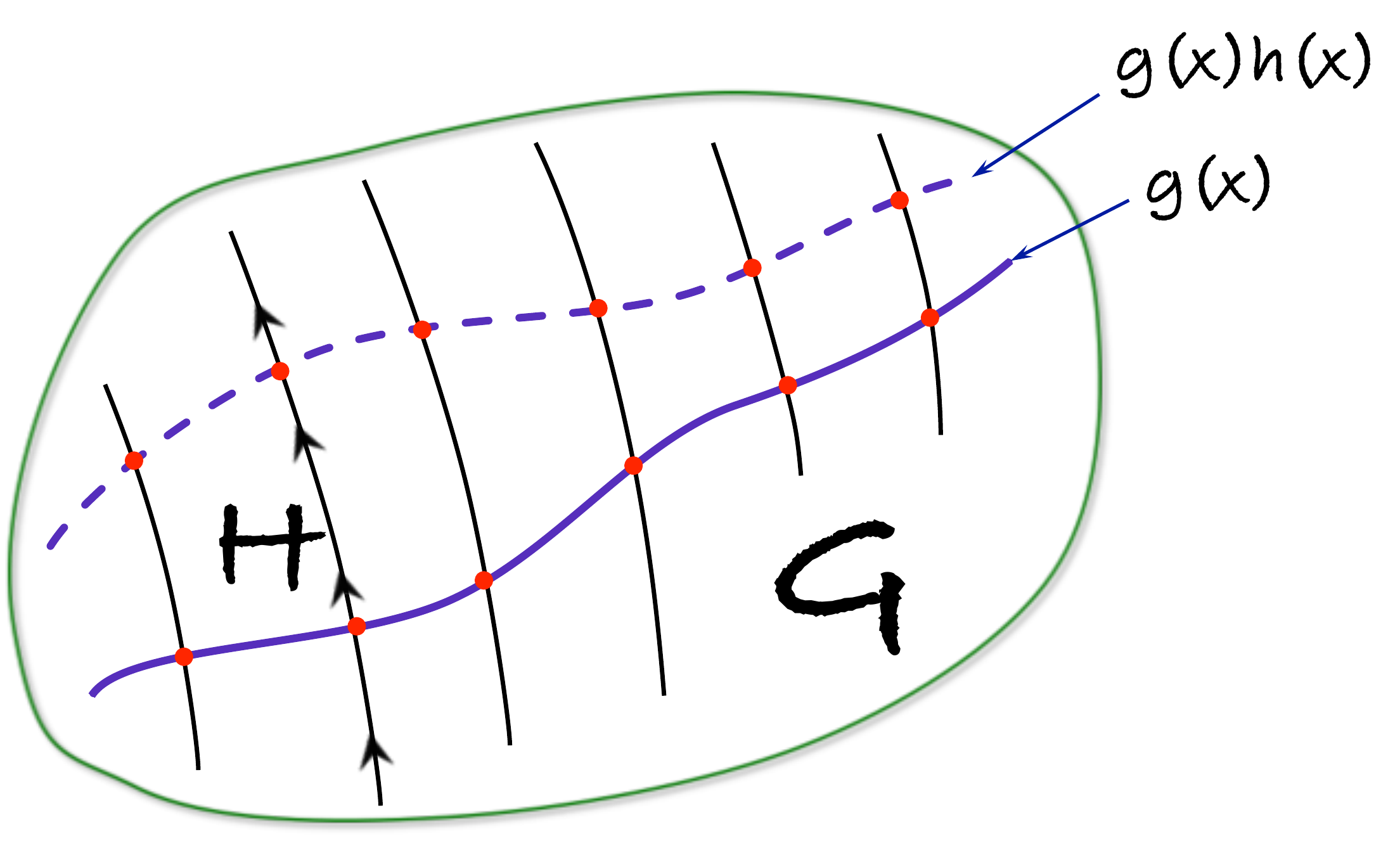}}
\caption{\label{gauge}\small Locally, a coset can be viewed as a section of the group manifold transverse to the orbits of $H$.
Two sections  that differ by multiplication with elements in $H$ are equivalent.}
\end{center}
\end{figure}

Sigma-models on homogeneous spaces are also integrable, provided that the 
target-space is {\it symmetric}, an algebraic property described below. The action of a sigma-model can be always written as  (\ref{Lagrangian-sigma-general}) in an explicit coordinate system, but for homogeneous spaces an elegant invariant construction based on the coset representation \cite{Callan:1969sn} is  in many ways more convenient. Points in a coset  are orbits of $H$. Picking one representative per orbit  introduces an explicit coordinate system. This defines a section of the group manifold that is transverse to the orbits of $H$ at any point\footnote{Such a section may not be globally defined.}. A field configuration then is a map from space-time to the group manifold which is restricted to this section (fig.~\ref{gauge}).  But any other choice differing by right multiplication with elements in $H$ is equivalent, and it is desirable to have a formulation covariant with respect to the action of $H$. In other words, any map $g(x)$ should be allowed, but the field configurations $g(x)$ and $g(x)h(x)$ with $h(x)\in H$ should be physically indistinguishable. This can be achieved by realizing the right action of $H$ as a gauge symmetry:
\begin{equation}\label{locgauge}
 g(x)\rightarrow g(x)h(x),\qquad h(x)\in H.
\end{equation}
For this construction to work the action of the sigma-model has to be gauge-invariant. 

The Lagrangian (\ref{PCFLagrangian}) fails to satisfy this condition. Indeed, the current
(\ref{LIcurrent}) transforms as 
\begin{equation}
 j_a\rightarrow h^{-1}j_ah+h^{-1}\partial _ah,
\end{equation}
and the inhomogeneous term does not cancel in the gauge variation of the action. 

The inhomogeneous term in the variation of the current can be projected out by decomposing the Lie algebra $\mathfrak{g}$ of $G$ into the denominator subalgebra $\mathfrak{h}$ and its orthogonal complement $\mathfrak{f}$ with respect to the invariant quadratic form:
\begin{equation}\label{g=h+f}
 \mathfrak{g}=\mathfrak{h}\oplus\mathfrak{f}.
\end{equation}
The current then expands in two components:
\begin{equation}
 j_a= A_a+K_a,\qquad A_a\in\mathfrak{h},~K_a\in\mathfrak{f}.
\end{equation}
Since $h\in H$, and $h^{-1}\partial _ah\in\mathfrak{h}$,
it is $A_a$ that transforms as a gauge field:
\begin{equation}
 A_a\rightarrow h^{-1}A_ah+h^{-1}\partial _ah,
\end{equation}
while the coset component transforms as a matter field in the adjoint:
\begin{equation}
 K_a\rightarrow h^{-1}K_ah.
\end{equation}

The simplest gauge-invariant Lagrangian built from these fields is
\begin{equation}
 \mathcal{L}=-\mathop{\mathrm{tr}}K_aK^a.
\end{equation}
It is manifestly invariant under the local gauge transformations (\ref{locgauge}) and is also symmetric under global left multiplications $g(x)\rightarrow g_Lg(x)$, $g_L\in G$, because the current is left invariant from the very beginning. The model  possesses all the requisite symmetries of the coset space $G/H$.

In order to make the discussion less abstract let us illustrate the coset construction by a simple example of the Hopf fibration $S^3\rightarrow S^2$:

\begin{expl}
$S^2=SU(2)/U(1)$
\end{expl}
\begin{proof}
The generators of $\mathfrak{su}(2)$ (the Pauli matrices) are split into the coset generators and the denominator subalgebra as
\begin{equation}\label{Euler}
  \mathfrak{h}=<\sigma _1>,\qquad \mathfrak{f}=<\sigma _2,\sigma _3>.
\end{equation}
Generic element of $SU(2)$ can be parameterized by three Euler angles:
\begin{equation*}
  g=\,{\rm e}\,^{\frac{i\varphi \sigma _1}{2}}\,{\rm e}\,^{\frac{i\theta \sigma _3}{2}}
  \cancelto{\in H}{\,{\rm e}\,^{\frac{i\psi \sigma _1}{2}}}.
\end{equation*}
The last factor belongs to the coset denominator and can be removed by a gauge transformation. We can just drop it altogether, which is equivalent to fixing a gauge. For the current we then have:
$$
 g^{-1}dg=\,{\rm e}\,^{-\frac{i\theta \sigma _3}{2}}\left(d+\frac{i}{2}\,d\varphi \sigma _1\right)\,{\rm e}\,^{\frac{i\theta \sigma _3}{2}}=\frac{i}{2}\left(\underbrace{d\theta \sigma _3+d\varphi \sin\theta \sigma _2}_{K}+\underbrace{d\varphi \cos\theta \sigma _1}_{A}\right).
$$
The resulting line element is
$$
 ds^2=-2\mathop{\mathrm{tr}}K^2=
 \frac{1}{2}\,\mathop{\mathrm{tr}}\left(d\theta \sigma _3+d\varphi \sin\theta \sigma _2\right)^2
 =d\theta ^2+\sin^2\theta \,d\varphi ^2,
$$
which is the familiar metric of the sphere.
\end{proof}

Returning to the general case, we can re-write the variation of the current (\ref{varcurrent}) as
\begin{equation}\label{decvar}
  \delta j_a=D_a\xi +[K_a,\xi ],
\end{equation}
where $D_a$ is the standard covariant derivative: 
\begin{equation}
 D_a=\partial _a+[A_a,\cdot].
\end{equation}
The ensuing variation of the Lagrangian is
$$
 \delta \mathcal{L}=-2\mathop{\mathrm{tr}}\delta K_aK^a
 =-2\mathop{\mathrm{tr}}\delta j_aK^a
 =-2\mathop{\mathrm{tr}}D_a\xi K^a,
$$
where the first equality follows from orthogonality of $\delta A_a\in \mathfrak{h}$ and $K^a\in\mathfrak{f}$, and the commutator term in (\ref{decvar}) was dropped because of the  invariance of the quadratic form.  The resulting equations of motion are
\begin{equation}\label{DK}
 D_aK^a=0.
\end{equation}
Together with the flatness condition (\ref{flatcurrent}), they form a complete set of equations for the current components.
The flatness condition, when expressed in terms of $K_a$ and $A_a$, takes the form:
\begin{equation}\label{F_ab}
 F_{ab}+D_aK_b-D_bK_a+[K_a,K_b]=0,
\end{equation}
where $F_{ab}=\partial _aA_b-\partial _bA_a+[A_a,A_b]$ is the field strength of the gauge connection. 

Before proceeding further let us pause for a few general remarks on  the coset decomposition (\ref{g=h+f}). It is orthogonal with respect to the invariant scalar product on $\mathfrak{g}$, but what about the commutation relations? How are they affected by the coset decomposition?
Since $\mathfrak{h}$ is a subalgebra, the commutator of two elements of $\mathfrak{h}$ is again an element of $\mathfrak{h}$:
\begin{equation}\label{hhf}
 [\mathfrak{h},\mathfrak{h}]\subset\mathfrak{h}.
\end{equation}
In terms of the structure constraints this means that $f_{\mathfrak{h}\mathfrak{h}\mathfrak{f}}=0$. Anti-symmetry of the structure constants (a very general property that follows from the existence of a non-degenerate invariant quadratic form) implies that $f_{\mathfrak{h}\mathfrak{f}\mathfrak{h}}=0$. Hence,
\begin{equation}\label{hff}
 [\mathfrak{h},\mathfrak{f}]\subset\mathfrak{f}.
\end{equation}

What can be said about $[\mathfrak{f},\mathfrak{f}]$?
Since $f_{\mathfrak{f}\mathfrak{f}\mathfrak{h}}$ are related to $f_{\mathfrak{h}\mathfrak{f}\mathfrak{f}}$, they are different from zero. The structure constants $f_{\mathfrak{f}\mathfrak{f}\mathfrak{f}}$ are not restricted by any general principles. An interesting special cases arises if they also vanish. Then
\begin{equation}\label{ffh}
 [\mathfrak{f},\mathfrak{f}]\subset\mathfrak{h}.
\end{equation}
If the Lie algebra decomposition obeys this constraint, the coset space is called symmetric.
The commutation relations (\ref{hhf}), (\ref{hff}) and (\ref{ffh}) follow a systematic pattern in this case. One can formalize it by introducing 
 a  $\mathbbm{Z}_2$ transformation acting on the Lie algebra $\mathfrak{g}$ as
\begin{equation}\label{Z2}
 \Omega (\mathfrak{h})=\mathfrak{h}, \qquad 
 \Omega (\mathfrak{f})=-\mathfrak{f}.
\end{equation}
This operation constitutes a symmetry of  $\mathfrak{g}$ once (\ref{ffh}) holds. Compatibility of the $\mathbbm{Z}_2$ symmetry (\ref{Z2}) with the commutation relations of $\mathfrak{g}$ can be taken as a definition of symmetric homogeneous space.

One can invert the logic and think of a symmetric space as being defined by 
 the symmetry group $G$ and a $\mathbbm{Z}_2$ automorphism of its Lie algebra (a linear map $\Omega :\mathfrak{g}\rightarrow \mathfrak{g}$ that preserves the commutation relations and squares to the identity: $\Omega ^2={\rm id}$). The denominator of the coset $\mathfrak{h}$ is then identified with the invariant subalgebra of $\Omega $. The $\mathbbm{Z}_2$ automorphism of the Lie algebra translates into a geometric reflection symmetry of the homogeneous space, as illustrated by the following example.

\begin{expl}
$S^2=SU(2)/U(1)$ is a symmetric space.
\end{expl}
\begin{proof}
 The coset decomposition (\ref{Euler}) is $\mathbbm{Z}_2$ invariant because $[\sigma _2,\sigma _3]=2i\sigma _1\in\mathfrak{h}$. Since $\sigma ^3$ pairs with $\theta $ in the Euler parameterization,
 the $\mathbbm{Z}_2$ symmetry of the algebra translates to the invariance of the sphere under $\theta \rightarrow -\theta $.
\end{proof}

\begin{exc}
 Check that $S^n$, $\mathbbm{CP}^n$ and $(E)AdS_{d+1}$ are symmetric spaces.
\end{exc}

\begin{exc}
The conformal algebra $\mathfrak{so}(d,2)$ is generated by Lorentz transformations $L_{\mu \nu }$, translations $P_\mu $, dilatation $D$ and special conformal transformations $K_\mu $. Its commutation relations are
\begin{eqnarray}\label{coalg}
 &&[L_{\mu \nu },L_{\lambda \rho }]=\eta _{\mu \lambda }L_{\nu \rho }+\eta
 _{\nu \rho }L_{\mu \lambda }-\eta _{\mu \rho }L_{\nu \lambda }-\eta
 _{\nu \lambda }L_{\mu \rho },
 \nonumber \\
&& {}[L_{\mu \nu },P_\lambda ]=\eta _{\mu \lambda }P_\nu -\eta _{\nu \lambda}
 P_\mu, \qquad 
{}[   L_{\mu \nu },K_\lambda ]=\eta _{\mu \lambda }K_\nu -\eta _{\nu
 \lambda }K_\mu,
\nonumber \\
 &&{}[ D,P_\mu ]=P_\mu,\qquad 
 {}[ D,K_\mu ]=-K_\mu,\qquad 
 {}[ P_\mu ,K_\nu ]=2L_{\mu \nu }-2\eta _{\mu \nu }D.
\end{eqnarray}
The Killing metric on $\mathfrak{so}(d,2)$ is
\begin{equation}
 \mathop{\mathrm{tr}}L_{\mu \nu }L_{\lambda \rho }=\eta _{\mu [\lambda  }\eta _{\nu \rho ]},
 \qquad 
 \mathop{\mathrm{tr}}D^2=-1,
 \qquad 
 \mathop{\mathrm{tr}}P_\mu K_\nu =-2\eta _{\mu \nu }.
\end{equation}

Consider a $\mathbbm{Z}_2$ transformation: 
\begin{equation}
 \Omega (L_{\mu \nu })=L_{\mu \nu },
 \qquad 
 \Omega (D)=-D,
 \qquad 
 \Omega (K_\mu )=P_\mu ,
 \qquad 
 \Omega (P_\mu )=K_\mu .
\end{equation}
Check that $\Omega $ is consistent with the commutation relations of $\mathfrak{so}(4,2)$, and therefore constitutes an automorphism of the algebra. Show that the invariant subalgebra of $\Omega $ is $\mathfrak{so}(d,1)$, and the symmetric coset defined by $\Omega $ is $AdS_{d+1}$.

Take the coset representative to be
\begin{equation}
 g(x,z)=\,{\rm e}\,^{iP_\mu x^\mu }z^{-D},
\end{equation}
and derive the metric in the $(x^\mu ,z)$ coordinates from the coset construction.
\end{exc}

A remarkable fact about symmetric-space sigma models is their complete integrability \cite{Eichenherr:1979ci}. While there coset sigma-models which are not symmetric but still integrable \cite{Bykov:2016rdv}, integrability of symmetric cosets can be demonstrated in a uniform, essentially algebraic way. 

Here is the key point, that singles out symmetric spaces.  The first two terms in (\ref{F_ab}) belong to $\mathfrak{h}$ and $\mathfrak{f}$, respectively, because of (\ref{hhf}) and (\ref{hff}), while the third term in general has both $\mathfrak{f}$ and $\mathfrak{h}$ components. But for symmetric cosets  all four terms in the flatness condition (\ref{F_ab}) have definite $\mathbbm{Z}_2$ grading and the equation can therefore be neatly projected onto $\mathfrak{h}$ and $\mathfrak{f}$ subspaces of the coset decomposition:
\begin{eqnarray}\label{symmcos-flat}
&& F_{ab}+[K_a,K_b]=0
\nonumber \\
&& D_aK_b-D_bK_a=0.
\end{eqnarray}
These equations, together with the equation of motion (\ref{DK}) follow from the flatness condition for the Lax connection:
\begin{equation}\label{Laxsymmc}
 L_a=A_a+\frac{z^2+1}{z^2-1}\,K_a-\frac{2z}{z^2-1}\,\varepsilon _{ab}K^b.
\end{equation}

\begin{exc}
 Show that the zero-curvature condition (\ref{flatc}) for  (\ref{Laxsymmc}) is equivalent to the equations of motion (\ref{DK}), (\ref{symmcos-flat}).
\end{exc}

\section{B-field and topology}\label{B-field:sec}

So far we have discussed the metric term in the sigma-model Lagrangian (\ref{Lagrangian-sigma-general}), and showed how to construct it for manifolds that admit an action of a symmetry group $G$. Is it also possible to construct an invariant B-field? To better formulate the question, it is convenient to switch to the form notations:
\begin{equation}
 \frac{1}{2}\int_{\Sigma }^{}d^2x\,\varepsilon ^{ab}B_{MN}\partial _aX^M\partial _bX^N=\int_{\Sigma }^{}B=\int_{D}^{}H,
\end{equation}
where
\begin{equation}
B=B_{MN}dX^M\wedge dX^N,~
 H=dB= H_{MNL}dX^M\wedge dX^N\wedge dX^L,~
 \label{closedH}
 H_{MNL}=\partial _{[M}B_{NL]},
\end{equation}
and $D$ is a fiducial three-dimensional space whose boundary is $\Sigma $: $\partial D=\Sigma $. 

Re-writing a 2d integral in a  3d form may look superficial, but it helps to uncover interesting possibilities that would be overlooked otherwise. The three-form $H$, as defined in (\ref{closedH}), is exact. Interestingly, this requirement can be relaxed and replaced by a weaker condition that   the form is closed ($dH=0$). The three-dimensional form of the B-field coupling cannot  be then  transformed back to a two-dimensional integral of an invariant two-form. Yet it defines a consistent local coupling in two-dimensional field theory \cite{Novikov:1982ei,Witten:1983tw}.

If $\Sigma =S^2$, then $D$ is the three-dimensional ball. Field configurations on $S^2$  can be unfolded to $\mathbbm{R}^2$ by stereographic projection and, conversely, field configurations on $\mathbbm{R}^2$ sufficiently well behaved at infinity can be compactified on $S^2$, so the ensuing construction applies to conventional field theories on the flat 2d space under assumption of regularity conditions at infinity. 

The B-field coupling in the sigma-model action can thus be constructed from an invariant two-form on the target space or from an invariant closed three-form. Both cases lead to topological terms in the sigma-model action.
We consider them in turn. A more comprehensive introduction to topological terms in QFT can be found in \cite{Abanov:2017zok}.

\subsection{Theta term}

If the coset denominator contains an Abelian factor, $H\ni U(1)$, the corresponding field strength $F_{ab}$ is a gauge-invariant two-form. The coset sigma-model on $G/H$ then admits a theta-term:
\begin{equation}
 S_\theta =\frac{\theta }{4\pi }\int_{}^{}d^2x\,\varepsilon ^{ab}F_{ab}
 =\frac{\theta }{2\pi }\int_{}^{}F.
\end{equation}
The theta-term is topological, it only depends on the values of the fields at infinity:
$$
 \frac{1}{2\pi }\int_{}^{}F=\frac{1 }{2\pi }\oint_\infty A.
$$
Since the field must approach pure gauge  at $x\rightarrow \infty $: $A=d\varphi $, 
$$
 \frac{1 }{2\pi }\oint_\infty A=\frac{1 }{2\pi }\oint_\infty d\varphi 
 =\frac{\Delta \varphi }{2\pi }=n.
$$

The theta-term is non-zero only on topologically non-trivial field configurations (instantons) and measures their topological charge. As such it does not contribute to the equations of motion. In the path integral the theta-term produces a phase factor $\,{\rm e}\,^{iS_\theta }=\,{\rm e}\,^{i\theta n}$. Since $n$ is an integer, the theta-angle is a periodic variable: the path integral does not change under $\theta \rightarrow \theta +2\pi $.

Examples of sigma-models that admit a theta-term, and have instantons, are the $\mathbbm{CP}^n$ models, for any $n$. These include the $O(3)$ model (because  $S^2=\mathbbm{CP}^1$), while $O(n)$ models with $n>3$ do not admit a theta term and do not have instantons. 

The theta-term measures the winding of pure gauge field at infinity and is associated with the homotopy group\footnote{The definition and properties of homotopy groups can be found, for instance, in \cite{Dubrovin-book}.} $\pi _1(H)=\mathbbm{Z}$. If $\pi _2(G/H)=\mathbbm{Z}$, the degree of mapping $\Sigma \rightarrow G/H$ also defines a topological charge, but one can show on very general grounds that $\pi _2(G/H)=\pi _1(H)$ and the degree of mapping  actually coincides with the winding of the pure gauge  at infinity, so we do not get anything new. It is instructive to explicitly work out the equivalence of the two definitions of the topological charge for the $O(3)$ model.

\begin{exc}
 The $O(3)$ sigma model can be realized as the $SU(2)/U(1)=S^2$ coset, or as a field theory of three-dimensional unit vector  $\mathbf{n}$,  $\mathbf{n}^2=1$. The degree of  mapping defined by $\mathbf{n}$ (the number of times the target-space $S^2$ wraps $\Sigma$) is given by the integral
 $$
  n=\frac{1}{8\pi }\int_{ }^{}d^2x\,\varepsilon ^{ab}\mathbf{n}\cdot \partial _a\mathbf{n}\times \partial _b\mathbf{n}.
 $$
 Prove that this number coincides with the instanton charge in the coset formulation, perhaps up to a sign. (Hint: use $F=-K\wedge K$, which is an identity, not an equation of motion).
 
 Instanton solutions of arbitrary topological charge in the $O(3)$ model have been constructed in \cite{Polyakov:1975yp}. An interplay between instantons and integrability in this and other two-dimensional models remains an interesting open problem. 
\end{exc}

\subsection{Wess-Zumino term}

\begin{figure}[t]
\begin{center}
 \subfigure[]{
   \includegraphics[height=5.2 cm] {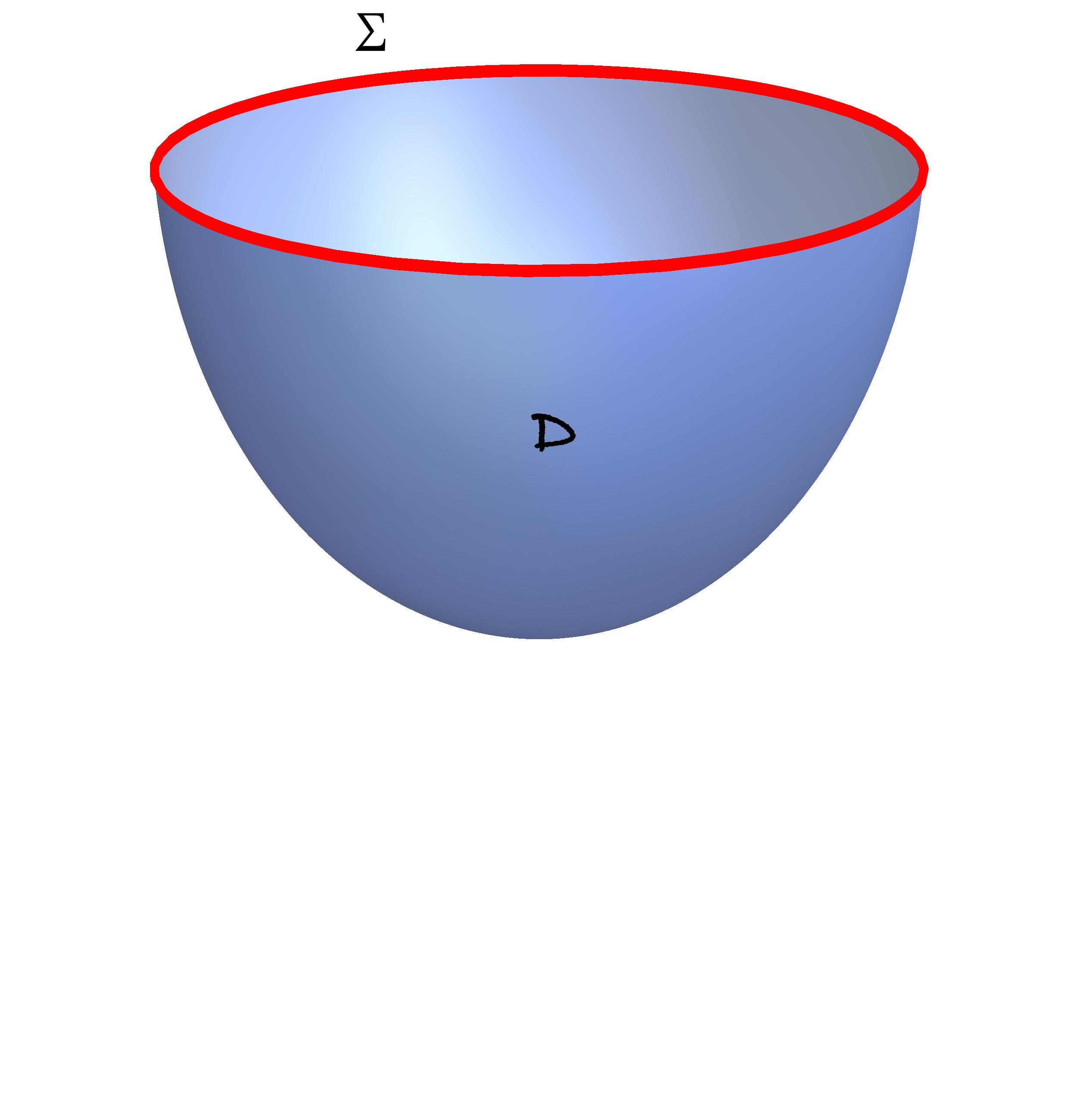}
   \label{fig3:subfig1}
 }
 \subfigure[]{
   \includegraphics[height=5.2 cm] {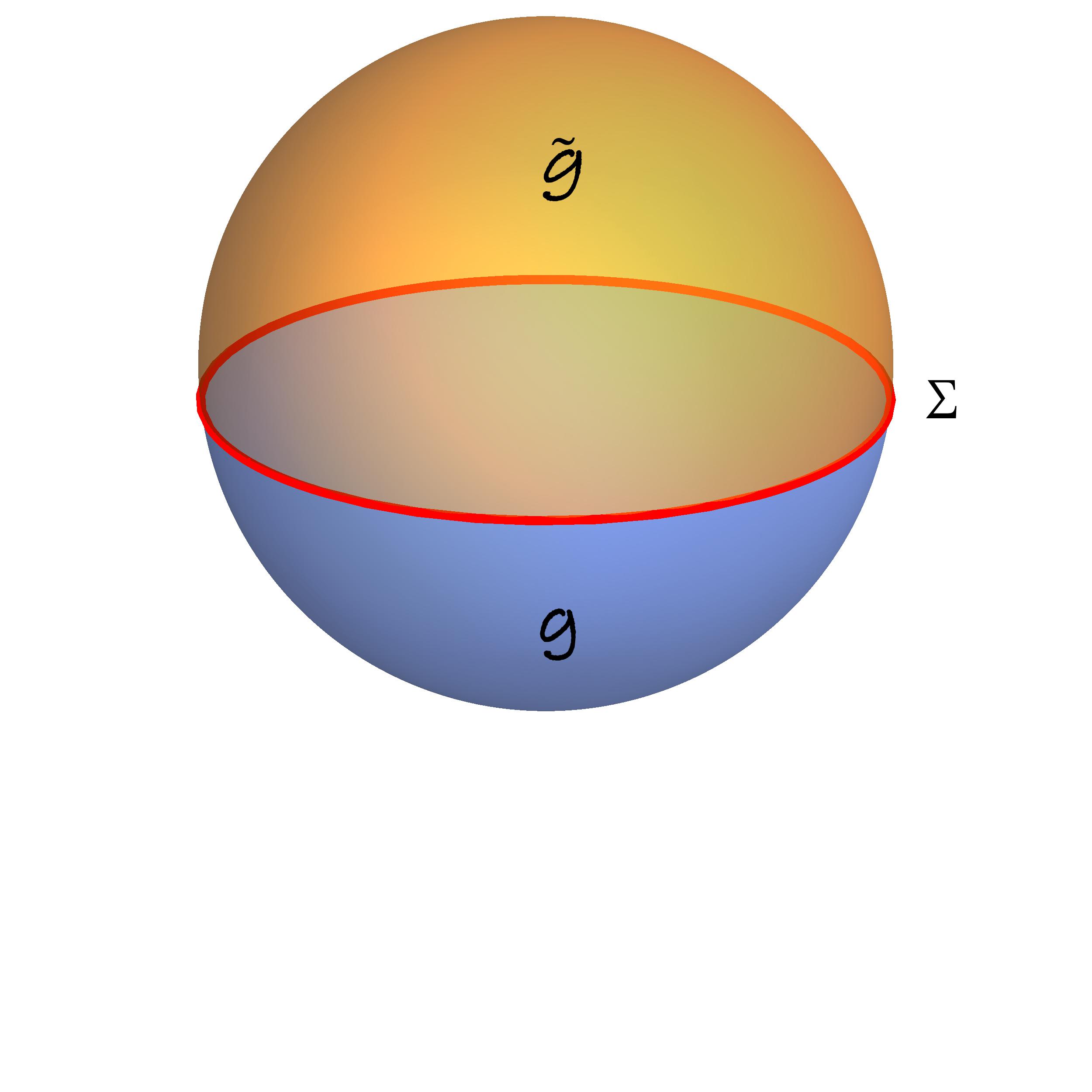}
   \label{fig3:subfig2}
 }
\caption{\label{WZy}\small (a) The space-time manifold $\Sigma $  is the boundary of the integration domain $D$ in the WZ term. (b) Two different continuations of the field variable in the interior of $D$ can be combined into a map $\hat{g}:\,S^3\rightarrow G$.}
\end{center}
\end{figure}

An invariant three-form on a group manifold can be obtained by taking the wedge product of three currents. This defines the Wess-Zumino (WZ) term that can be added to the action of  the principal chiral field sigma-model  \cite{Witten:1983ar}:
\begin{equation}
 S_{\rm WZ}=\frac{1}{3}\int_{D}^{}\mathop{\mathrm{tr}}j\wedge j\wedge j.
\end{equation}
The boundary of the integration domain is the physical 2d space-time: $\partial D=\Sigma $, fig.~\ref{fig3:subfig1}. The definition tacitly assumes that the field variable $g(x)$ is continued into the interior of $D$ in some way. This is always possible to do, since $\pi _2(G)=0$ for any Lie group $G$ and consequently any map $g:\,S^2\rightarrow G$ is contractible. The process of contraction defines a  map  $D\rightarrow G$. 

If considered as part of the action functional in a 2d field theory, the WZ term should be independent of the way the field variable is continued inside $D$.
Let us see if this is indeed the case \cite{Witten:1983tw,Witten:1983ar}, by considering the WZ term evaluated on two different maps $g:\, D\rightarrow G$ and $\tilde{g}:\,D\rightarrow G$, which coincide on the boundary, $\left.g\right|_\Sigma =\left.\tilde{g}\right|_\Sigma$ and thus correspond to one and the same field configuration in space-time. Consider first infinitesimally close  maps: $\tilde{g}=g(1+\xi )$, where $\xi \in\mathfrak{g}$ is small. Substituting the variation of the current from (\ref{varcurrent})  in the WZ functional we find:
\begin{equation}
 \delta S_{\rm WZ}=\int_{D}^{}d\xi \wedge j\wedge j=
 -\int_{D}^{}\xi \,d\left(j\wedge j\right).
\end{equation}
The boundary term can be omitted upon integration by parts because $\xi $ vanishes on $\Sigma $. The flatness condition (\ref{wedge}) then gives
$$
 d\left(j\wedge j\right)=dj\wedge j-j\wedge dj=-j\wedge j\wedge j+j\wedge j\wedge j=0,
$$
so the variation of the WZ action  indeed vanishes.

Continuous deformations of the map $g$ do not change the WZ term, but can any two maps from $D$ to the group manifold be  continuously connected?  Consider two such maps $g$ and $\tilde{g}$ and combine them into a map $\hat{g}:\,S^3\rightarrow G$ defined as shown in fig.~\ref{fig3:subfig2}:
$$
 \hat{g} (x)=
\begin{cases}
 \tilde{g}(x) & x\in{\rm Upper~Hemisphere}
\\
  g(x) & x\in{\rm Lower~Hemisphere}.
\end{cases}
$$
Then
\begin{equation}\label{deltaSWZ}
  S_{\rm WZ}[\tilde{g}]-S_{\rm WZ}[g]=\frac{1}{3}\int_{S^3}^{}
  \mathop{\mathrm{tr}}\hat{j}\wedge\hat{j}\wedge
\hat{j}
=16\pi ^2in.
\end{equation}
The number $n$ does not change under continuous deformations of the map $\hat{g}$ and therefore is a topological invariant. Under appropriate normalization of the quadratic form on the Lie algebra this number is an integer that characterizes an element of $\pi _3(G)=\mathbbm{Z}$.

\begin{exc}
For $G=SU(2)=S^3$, $\hat{g}$ defines a map $S^3\rightarrow S^3$. Show that $n$ given by the integral (\ref{deltaSWZ}) is the degree of this map, the number of times $\hat{g}(x)$ wraps the target $S^3$ while $x$ goes around the sphere.
\end{exc}

The WZ action thus is a multi-valued functional. Depending on how the field variable $g$ is continued inside the three-dimensional domain of integration the result may change by an integer:
\begin{equation}
 S_{\rm WZ}\sim S_{\rm WZ}+16\pi ^2in.
\end{equation}
This is not a problem in classical field theory, because the equations of motion only depend on the variation of the action, and the variation of the WZ action is single-valued. But in quantum theory the action enters the path integral through the phase factor
$$
 \,{\rm e}\,^{\frac{k}{8\pi }\,S_{\rm WZ}}\in{\rm path~integral},
$$
where $k$ is the coupling constant (called level). This factor will be single-valued provided that the coupling is quantized:
\begin{equation}
 k\in\mathbbm{Z}.
\end{equation}

The principal chiral field with the WZ term, the Wess-Zumino-Witten (WZW) model, is defined by the action \cite{Witten:1983ar}
\begin{equation}\label{WZW}
 S_{\rm WZW}=-\frac{1}{\kappa ^2}\int_{\Sigma }^{}\mathop{\mathrm{tr}}j\wedge *j
 +\frac{ik}{24\pi ^2}\int_{D}^{}\mathop{\mathrm{tr}}j\wedge j\wedge j.
\end{equation}

\begin{exc}
Derive the equations of motion for the WZW model and show that they admit a Lax representation \cite{TakhtajanVeselov84}. What happens at $\kappa ^2=2\pi /k$?
\end{exc}

\section{Quantum sigma-models}\label{QSM:sec}

A quantum coset sigma-model is defined by the path integral (in Euclidean signature)
\begin{equation}
 Z=\int_{}^{}\mathcal{D}g\,\,{\rm e}\,^{\frac{1}{\kappa ^2}\int_{}^{}
 d^2x\,\mathop{\mathrm{tr}}K_aK^a}.
\end{equation}
This is a highly nonlinear quantum field theory with a single dimensionless coupling. To get a first glance at the properties of this QFT we consider its response to a macroscopic external force, assuming momentarily that residual quantum effects can be treated perturbatively.

To this end we separate the integration variable in the path integral into the classical background field $\bar{g}$, which models the external force and which we assume to be a solution of the classical equations of motion, and quantum field $X$ whose fluctuations are assumed to be small:
\begin{equation}\label{backgroundexp}
 g=\bar{g}\,{\rm e}\,^{X}.
\end{equation}
The field $X$ belongs to the Lie algebra of $G$. An appropriate gauge transformation (\ref{locgauge}) can eliminate the $\mathfrak{h}$ component of $X$:
\begin{equation}
 X\in\mathfrak{f}.
\end{equation}
The restriction to the coset subspace is a convenient gauge-fixing condition.

To develop perturbation theory we need to expand the current  (\ref{LIcurrent}) in powers of $X$. We start by deriving a general formula for such an expansion. Substituting (\ref{backgroundexp}) into the definition of the current we get:
\begin{equation}\label{jexp}
 j_a=\,{\rm e}\,^{-X}\left(\partial _a+\bar{j}_a\right)\,{\rm e}\,^{X}
 =\bar{j}_a+\int_{0}^{1}ds\,\,{\rm e}\,^{-sX}\mathcal{D}_aX\,{\rm e}\,^{sX},
\end{equation}
where $\bar{j}_a=\bar{g}^{-1}\partial _a\bar{g}$ is the background current and $\mathcal{D}_a=\partial _a+[\bar{j}_a,\cdot ]$. The last term expands in multiple commutators of $X$. To write the series in a compact form, we commute  both sides of (\ref{jexp}) with $X$, which eliminates the integral over $s$:
\begin{equation}
 [X,j_a-\bar{j}_a]=-\int_{0}^{1}ds\,\,\frac{d}{ds}\,\,{\rm e}\,^{-sX}\mathcal{D}_aX\,{\rm e}\,^{sX}
 =\mathcal{D}_aX-\,{\rm e}\,^{-X}\mathcal{D}_aX\,{\rm e}\,^{X}.
\end{equation}
Introducing the notation
\begin{equation}
 \mathop{\mathrm{ad}}X\cdot Y=[X,Y],
\end{equation}
and taking into account that
$$
 \,{\rm e}\,^{-X}Y\,{\rm e}\,^{X}=\,{\rm e}\,^{-\mathop{\mathrm{ad}}X}Y,
$$
we find:
\begin{equation}
 j_a=\bar{j}_a+\frac{1-\,{\rm e}\,^{-\mathop{\mathrm{ad}}X}}{\mathop{\mathrm{ad}}X}\,\mathcal{D}_aX.
\end{equation}
Using this formula the current can be expanded to any desired order in $X$.

After rescaling $X\rightarrow \kappa X$ the first two orders become
\begin{eqnarray}
  j_a&\simeq &\bar{j}_a+\kappa \mathcal{D}_aX-\frac{\kappa ^2}{2}\,[X,\mathcal{D}_aX]
  \nonumber \\
  &=&{\color{red}\bar{A}_a}+{\color{blue}\bar{K}_a}
  +\kappa \left({\color{blue}D_aX}+{\color{red}[\bar{K}_a,X]}\right)
  -\frac{\kappa ^2}{2}\left({\color{red}[X,D_aX]}+{\color{blue}[X,[\bar{K}_a,X]]}\right),
\end{eqnarray}
where $D_a$ is the covariant derivative with respect to the background gauge field. The coloring in the second line refers to the coset decomposition (\ref{g=h+f}). The coset component ($K_a$) is highlighted in blue. We have used the $\mathbbm{Z}_2$ grading to assign commutators to either $\mathfrak{h}$ or $\mathfrak{f}$, and thus  already here assume that the coset is a symmetric space.

The action  of the sigma-model, to the leading order in $\kappa $, is
\begin{equation}
 S=S_{\rm cl}-\int_{}^{}d^2x\,\mathop{\mathrm{tr}}\left(
 D_aXD^aX-[\bar{K}_a,X][\bar{K}^a,X]
 \right)+\mathcal{O}\left(\kappa \right).
\end{equation}
Gaussian integration over $X$ then yields:
\begin{equation}
 Z\simeq \,{\rm e}\,^{-S_{\rm cl}}\det\nolimits^{-\frac{1}{2}}\left(-D^2+\mathop{\mathrm{ad}}\bar{K}^2\right).
\end{equation}
The result of these manipulations is a one-loop effective action for the background field:
\begin{equation}
 S_{\rm eff}=S_{\rm cl}+\frac{1}{2}\,\ln\det\left(-D^2+\mathop{\mathrm{ad}}\bar{K}^2\right).
\end{equation}

The one-loop determinant diverges and requires regularization. Since in two dimensions the only scalar diagram that diverges is the bubble, the divergent part of the effective action is
\begin{equation}
   \parbox{1.5 cm}{\includegraphics[height=1.5 cm] {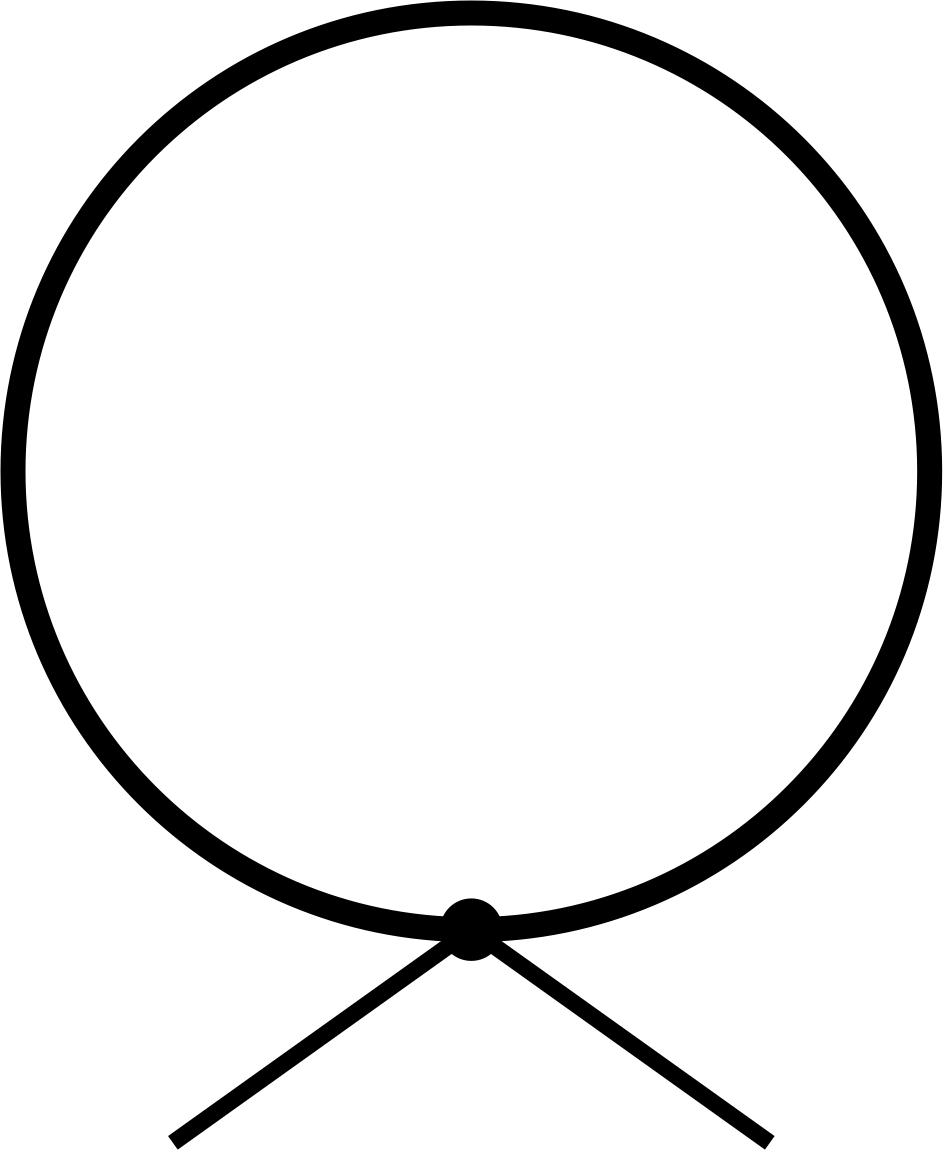}}
   =\frac{1}{2}\int_{}^{}d^2x\,\mathop{\mathrm{tr}}\nolimits_\mathfrak{f}\mathop{\mathrm{ad}}\bar{K}^2
   \times 
   \int_{}^{\Lambda }\frac{d^2p}{\left(2\pi \right)^2}\,\,\frac{1}{p^2}
  =\frac{1}{4\pi }\,\ln\frac{\Lambda }{\mu }\,
  \int_{}^{}d^2x\,\mathop{\mathrm{tr}}\nolimits_\mathfrak{f}\mathop{\mathrm{ad}}\bar{K}^2.
\end{equation}
As a matrix in the adjoint representation, $\mathop{\mathrm{ad}}\bar{K}_a$ has $\mathbbm{Z}_2$ grading $-1$ and thus the following matrix structure:
\begin{equation}
 \mathop{\mathrm{ad}}\bar{K}_a:~\begin{pmatrix}
  0 & * \\ 
  * & 0  \\ 
 \end{pmatrix}
 \begin{pmatrix}
  \mathfrak{h} \\ 
  \mathfrak{f} \\ 
 \end{pmatrix}.
\end{equation}
 So,
 $$
  \mathop{\mathrm{tr}}\nolimits_\mathfrak{f}\mathop{\mathrm{ad}}\bar{K}^2
  =\left\langle \mathfrak{f}\right|\mathop{\mathrm{ad}}\bar{K}_a\left|\mathfrak{h}\right\rangle\left\langle \mathfrak{h}\right|\mathop{\mathrm{ad}}\bar{K}^a\left|\mathfrak{f}\right\rangle
  =\mathop{\mathrm{tr}}\nolimits_\mathfrak{h}\mathop{\mathrm{ad}}\bar{K}^2
  =\frac{1}{2}\,\mathop{\mathrm{tr}}\nolimits_\mathfrak{g}\mathop{\mathrm{ad}}\bar{K}^2.
 $$

The divergence renormalizes the sigma-model coupling constant:
\begin{equation}
 S_{\rm eff}=-\int_{}^{}d^2x\,
 \left(\frac{1}{\kappa ^2}\,\mathop{\mathrm{tr}}-\frac{1}{8\pi }\,\ln\frac{\Lambda }{\mu }\,\mathop{\mathrm{tr}}\nolimits_{\rm adj}\right)\bar{K}_a\bar{K}^a+{\rm finite}.
\end{equation}
Taking into account that
\begin{equation}\label{dualCoxeter}
  \mathop{\mathrm{tr}}\nolimits_{\rm adj} XY=2C_G\,\mathop{\mathrm{tr}}XY,
\end{equation}
where $C_G$ is the dual Coxeter number\footnote{$C_{SU(N)}=N$, $C_{SO(N)}=N-2$.} of the group $G$,
we find that the coupling renormalized at scale $\mu $ is related to the bare coupling by
\begin{equation}
 \frac{1}{\kappa ^2(\mu )}=\frac{1}{\kappa ^2(\Lambda )}-\frac{C_G}{4\pi }\,\ln\frac{\Lambda }{\mu }\,.
\end{equation}

The running of sigma-model coupling with the energy scale is governed by the beta-function, which can be read off from the calculation above:
\begin{equation}
 \beta (\kappa )=\frac{d\kappa(\mu ) }{d\ln \mu }=-\frac{C_G}{8\pi }\,\kappa ^3.
\end{equation}
The one-loop beta function is negative, which means that two-dimensional coset sigma-models are asymptotically free quantum field theories \cite{Polyakov:1975rr}.

\begin{exc}
Calculate the beta-function of the WZW model (\ref{WZW}) \cite{Witten:1983tw}.
\end{exc}

\begin{figure}[t]
\begin{center}
 \centerline{\includegraphics[width=6cm]{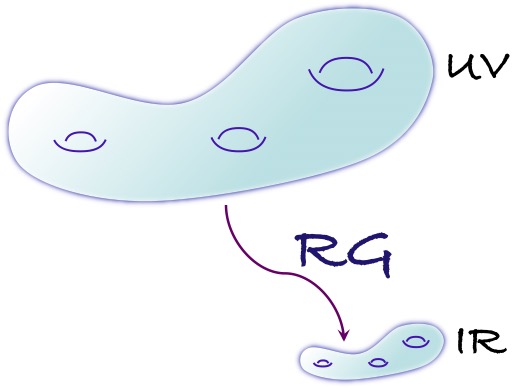}}
\caption{\label{rg}\small The RG flow in a sigma-model.}
\end{center}
\end{figure}

The RG flow in two-dimensional sigma-models is often interpreted geometrically as shrinking of the target-space in the IR \cite{Friedan:1980jf}. The coupling constant rescales the metric in the sigma-model action (\ref{Lagrangian-sigma-general}). Its inverse thus plays the role of an overall radius of the target space. The theory flows to stronger coupling in the IR, with the radius of the target space becoming smaller and smaller along the flow, fig.~\ref{rg}.

\begin{table}[t]
\caption{\small Examples of classically integrable sigma-models, which can be either integrable or not at the quantum level. }
\label{dtheories}
\begin{center}
\begin{tabular}[t]{c|c|c}
QFT & Coset & Integrability \\
\hline\hline
Principal Chiral Field & $G$ & \cmark \\
$O(N)$ model & $SO(N)/SO(N-1)$ & \cmark \\
$\mathbbm{CP}^N$ model & $SU(N+1)/SU(N)\times U(1)$ & \xmark \\
Superstring on $AdS_5\times S^5$ & $PSU(2,2|4)/SO(4,1)\times SO(5)$ & \cmark \\
$O(3)$ model with $\theta \neq 0,\pi $ & $SO(3)/SO(2)$ & \xmark
\end{tabular}
\end{center}
\end{table}

The RG flow will eventually stop and generate a dynamical mass scale through dimensional transmutation\footnote{In the presence of a theta-term or a WZ term this is not necessarily true. The flow may then end in a non-trivial CFT in the IR.}:
\begin{equation}\label{M-Lambda}
 M=\Lambda \,{\rm e}\,^{-\frac{4\pi }{C_G\kappa ^2}}.
\end{equation}
This means that coset sigma-models are massive field theories, with particles transforming under some representation of the symmetry group $G$. The dynamics of these particles is non-perturbative at low energies, but is highly constrained by integrability, of course under the assumption that integrability is not broken by quantization. 

Conservation laws are often affected by quantum anomalies, and integrable QFTs are no exceptions. Whether or not higher conserved charges are anomalous is a dynamical question, with a model-dependent answer. There are cases where integrability is preserved at the quantum level, and cases where integrability is broken by anomalies. Some representative examples are listed in table~\ref{dtheories}.

The dynamical information of a massive QFT is contained in its S-matrix. The S-matrix of an integrable QFT is so constrained by symmetries and analyticity that it can be reconstructed from the kinematical requirements alone, with the help of minimal amount of explicit perturbative data. We illustrate how such integrable bootstrap works taking the $O(N)$ model as an example. The departure point for constructing the exact scattering theory for the $O(N)$ model is its solution at large-$N$, that can be obtained by elementary manipulations of the path integral described in the next section. 

\section{$O(N)$ model: large $N$}\label{LargeN:sec}

The $O(N)$ model is a quantum field theory of an $N$-component vector field $\mathbf{n}$ subject to the constraint $\mathbf{n} ^2=1$. The constraint can be imposed by a Lagrange multiplier in the path integral:
\begin{equation}\label{O(N)pathintegral}
 Z=\int_{}^{} D\mathbf{n}\,D\sigma \,\,{\rm e}\,^{
 \frac{i}{\kappa ^2}\int_{}^{}d^2x\,
 \left[\left(\partial \mathbf{n}\right)^2-\sigma \left(\mathbf{n}^2-1\right)\right]
 }.
\end{equation}
Integration over $\sigma $ yields the requisite delta-function, so $D\mathbf{n}$ is the linear, unconstrained measure in the field space. Integrating over $\mathbf{n}$ first we get an effective action for $\sigma $:
\begin{equation}\label{sseff}
 S_{\rm eff}=\frac{1}{\kappa ^2}\int_{}^{}d^2x\,\sigma 
 +\frac{iN}{2}\,\ln\det\left(-\partial ^2-\sigma \right).
\end{equation}

The two terms are of the same order of magnitude in the 't~Hooft limit:
\begin{equation}
 N\rightarrow \infty ,\qquad \lambda =\frac{\kappa ^2N}{2}\,-\,{\rm fixed}.
\end{equation}
In the 't~Hooft limit the effective action is of an overall order $N$. The semiclassical approximation becomes exact at $N\rightarrow \infty $, and it suffices to solve the classical equations of motion that follow from (\ref{sseff}). Denoting
\begin{equation}
 \left\langle \sigma \right\rangle=m^2, 
 \end{equation}
 and taking the variation of (\ref{sseff}) we get:
\begin{equation}
 \frac{1}{\lambda }=i\left\langle x\vphantom{ \frac{1}{-\partial ^2-m^2 }}\right|
 \frac{1}{-\partial ^2-m^2 }
 \left| \vphantom{ \frac{1}{-\partial ^2-m^2 }}x\right\rangle
\qquad  \Longleftrightarrow \qquad 
 \frac{1}{\lambda }=\int_{}^{}\frac{d^2p}{\left(2\pi \right)^2}\,\,
 \frac{i}{p^2-m^2}
\end{equation}
This is known as the gap equation\footnote{The ground state of the system is of course homogeneous, but the gap equation also has space-time dependent soliton solutions that describe excited states, see \cite{Zarembo:2008hb} for an example.}, which after the Wick rotation, $p_0\rightarrow ip_0$ becomes
\begin{equation}
 \frac{1}{\lambda }=\int_{}^{}\frac{d^2p}{\left(2\pi \right)^2}\,\,
 \frac{1}{p^2+m^2}
 =\frac{1}{2\pi }\,\ln\frac{\Lambda }{m}\,.
\end{equation}
The UV cutoff $\Lambda $ is necessary for regulating the divergent momentum integral. The gap equation determines the mass scale of the theory:
\begin{equation}
 m=\Lambda \,{\rm e}\,^{-\frac{2\pi }{\lambda }}.
\end{equation}
This is consistent with (\ref{M-Lambda}) at large $N$, when the difference between $N$ and $C_{SO(N)}=N-2$ is immaterial.

If we return back to (\ref{O(N)pathintegral}), the expectation value of the Lagrange multiplier $\left\langle \sigma \right\rangle=m^2$ gives mass to the field $\mathbf{n}$, which thus describes $N$ massive particles in the vector representation of $SO(N)$. The $N-1$ Goldstone modes of the naive perturbation theory rearrange themselves into $N$ massive states at low energies. This rearrangement is obviously non-perturbative.

\begin{exc}\label{O(N)S-matrix:explicit}
Calculate the $2\rightarrow 2$ scattering amplitude in the $O(N)$ model to the first non-vanishing order in $1/N$  \cite{Zamolodchikov:1978xm}. 
\end{exc} 

\section{$O(N)$ model: exact solution}\label{ONexact:sec}

Integrability strongly constrains the scattering processes. The constraints follow from no more than simple kinematics.  Consider, to begin with, the scattering of two identical particles $A(p^{\rm in}_1)+A(p^{\rm in}_2)\rightarrow A(p^{\rm out}_1)+A(p^{\rm out}_2)$ in two dimensions, where $p_i^{\rm in}$, $p^{\rm out}_i$ are incoming and outgoing spacial momenta. If the initial momenta are given, energy and momentum conservation constitute two equations for two final-state momenta:
\begin{eqnarray}
 &&p^{\rm out}_1+p^{\rm out}_2=p^{\rm in}_1+p^{\rm in}_2
\nonumber \\
&& E(p^{\rm out}_1)+E(p^{\rm out}_2)=E(p^{\rm in}_1)+E(p^{\rm in}_2).
\nonumber 
\end{eqnarray}
Two equations for two unknowns have a discrete set of solutions -- two-body scattering has no phase space in two dimensions. 

It is easy to convince oneself that the conservation laws admit but two obvious solutions: $p^{\rm out}_1=p^{\rm in}_1$, $p^{\rm out}_2=p^{\rm in}_2$ (transmission) and $p^{\rm out}_1=p^{\rm in}_2$, $p^{\rm out}_2=p^{\rm in}_1$ (reflection). This is true in any translationally-invariant theory in two dimensions, but imagine now that there is an infinite number of additional conservation laws. In any interaction process, even if it involves more than two particles,  exchange of momenta would still be possible, but any other outcome of the scattering process would be kinematically forbidden, because conservation laws constitute a hugely overconstrained system of equations. In integrable theories, which feature an infinite number of conservation laws, particle production is forbidden, scattering is diagonal in particle number, and the final-state  momenta of any $n\rightarrow n$ process constitute a permutation of the initial-state momenta: $p_i^{\rm out}=p_{\sigma (i)}^{\rm in}$. Any permutation can be written as a product of transpositions, likewise $n\rightarrow n$ scattering in an integrable theory factorizes into a sequence of elementary $2\rightarrow 2$ processes. The dynamics of an integrable QFT thus is encoded in the two-body scattering S-matrix. This is a huge simplification that puts integrable models in between trivial (free) theories and generic QFTs where the intricate phase space of intermediate states makes the analytic structure of the S-matrix enormously complicated.

Taking into account that the $O(N)$ model is integrable, and drawing intuition from the large-$N$ approximation, we conclude that the model is characterized by
\begin{itemize}
 \item $N$ massive particles  transforming in the vector representation of $SO(N)$
 \item no particle production
 \item factorized scattering.
\end{itemize}
These basic data is sufficient to completely reconstruct the S-matrix of the theory. Our exposition of this bootstrap solution of the model closely follows the original derivation \cite{Zamolodchikov:1978xm}.

\begin{figure}[t]
\begin{center}
 \centerline{\includegraphics[width=6cm]{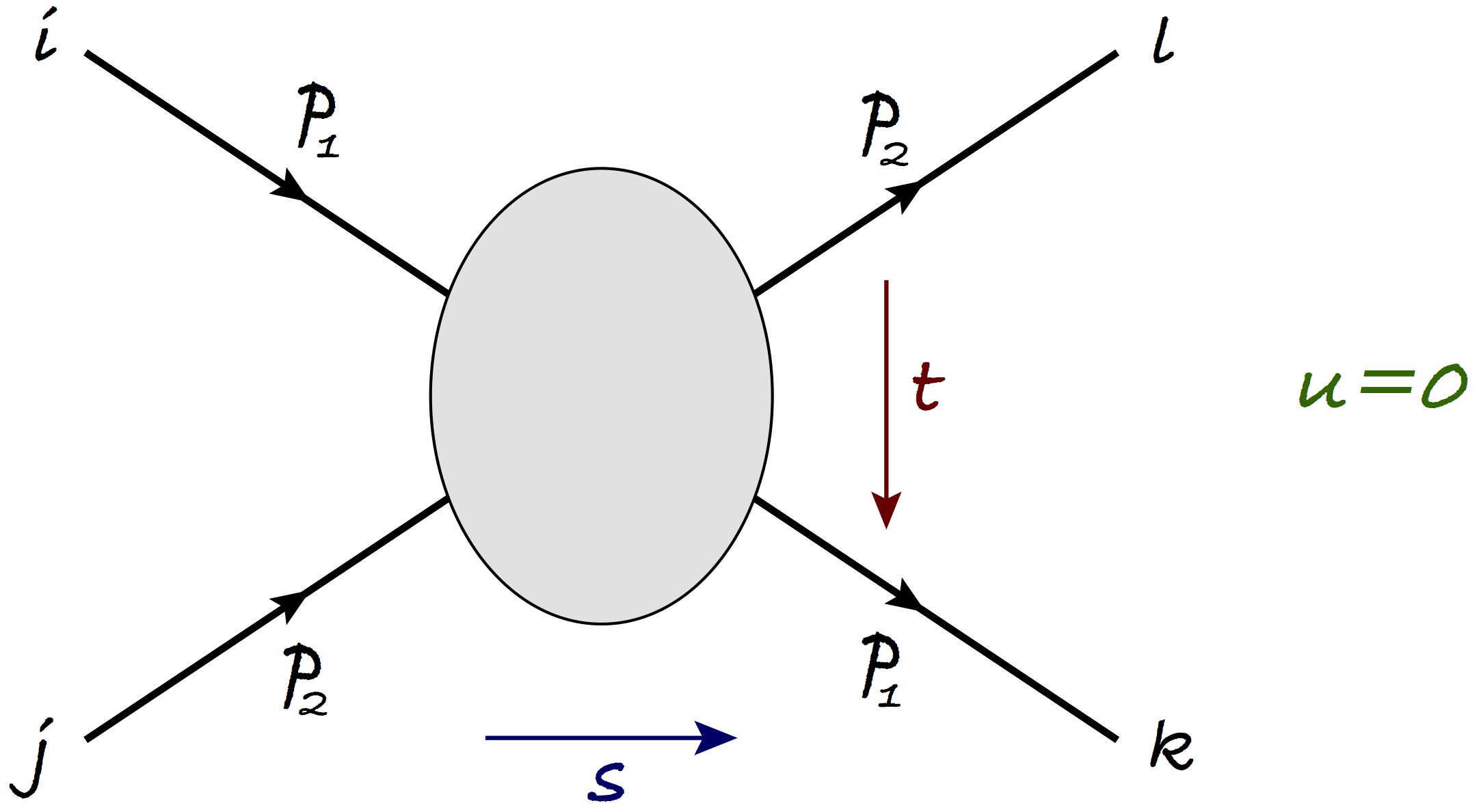}}
\caption{\label{s22}\small The $2\rightarrow 2$ scattering amplitude.}
\end{center}
\end{figure}

\begin{figure}[t]
\begin{center}
 \subfigure[]{
   \includegraphics[height=3.8 cm] {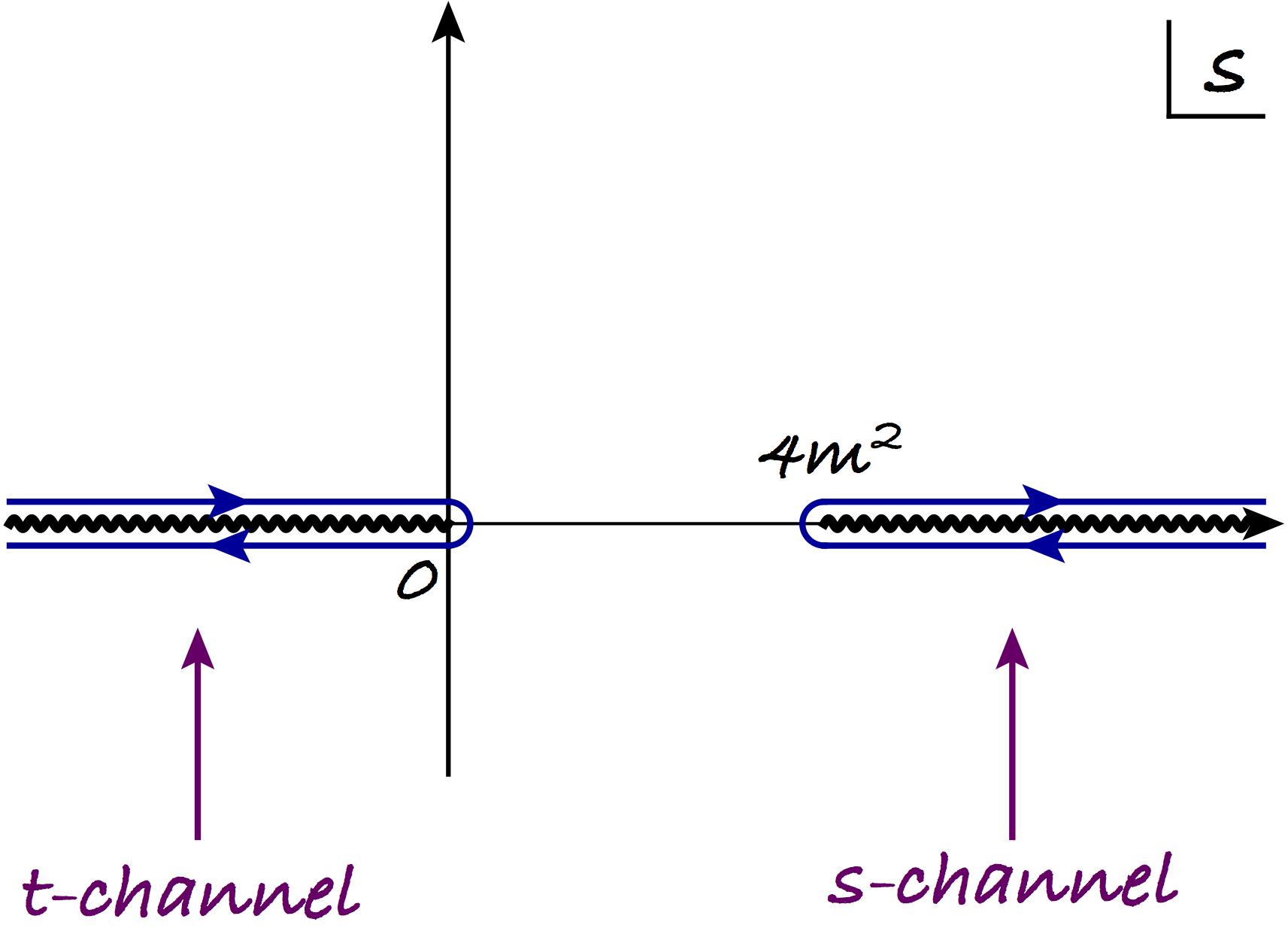}
   \label{fig4:subfig1}
 }
 \subfigure[]{
   \includegraphics[height=3.8 cm] {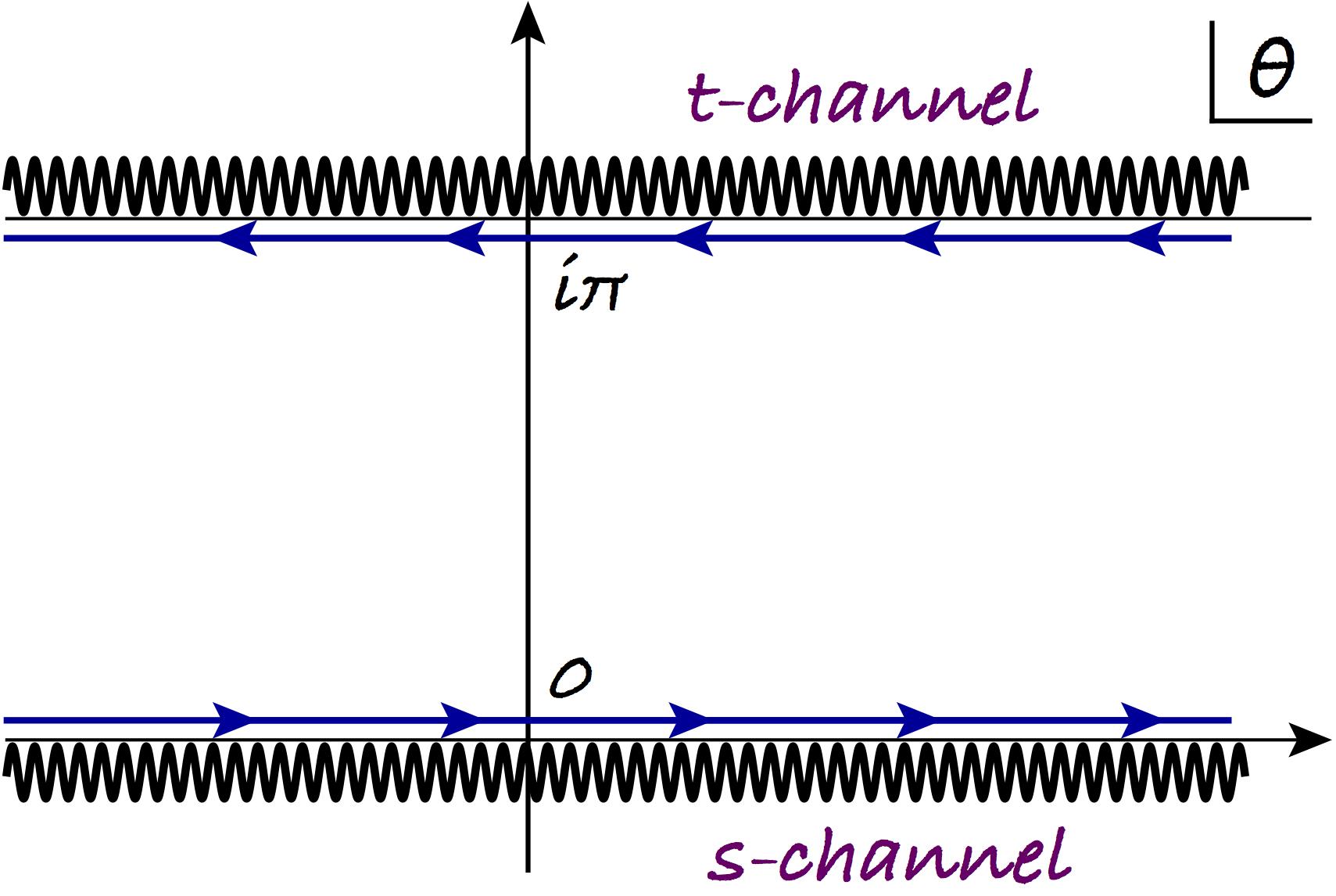}
   \label{fig4:subfig2}
 }
\caption{\label{planes}\small Analytic structure of the S-matrix: (a) in the $s$-plane; (b) in the rapidity plane.}
\end{center}
\end{figure}

Consider $2\rightarrow 2$ scattering illustrated in fig.~\ref{s22}. The Mandelstam invariant $u=0$ in two dimensions, while $s$ and $t$ are related by
\begin{equation}
 s+t=4m^2.
\end{equation}
Hence the $2\rightarrow 2$ S-matrix is a function of only one kinematic variable: $S_{ij}^{kl}(s)$. Moreover, the S-matrix is analytic in $s$ almost everywhere on the complex plane. All the singularities of the S-matrix are associated with intermediate virtual states going on-shell. And here integrability brings in enormous simplifications. The absence of particle production implies that the only allowed singularities are two-particle cuts in the $s$- and $t$-channels, fig.~\ref{fig4:subfig1}. The paucity of analytic functions on the complex plane with two cuts is the ultimate reason for the success of integrability bootstrap.

Two-dimensional relativistic kinematics drastically simplifies in the rapidity variables:
\begin{eqnarray}
 p_{1,2}&=&m\sinh\theta _{1,2}
\nonumber \\
E_{1,2}&=&m\cosh\theta _{1,2}.
\end{eqnarray}
The center-of-mass energy squared is then
\begin{equation}\label{srap}
 s=4m^2\cosh^2\frac{\theta_1-\theta _2 }{2}\,.
\end{equation}
The S-matrix is thus a function of the rapidity difference between the two colliding particles: $S=S(\theta _1-\theta _2)$. The entire $s$-plane maps onto the strip of the theta plane between the lines $\mathop{\mathrm{Im}}\theta =0$ and $\mathop{\mathrm{Im}}\theta =\pi $ (called the physical strip), as shown in fig.~\ref{planes}.

\begin{figure}[t]
\begin{center}
 \centerline{\includegraphics[width=6cm]{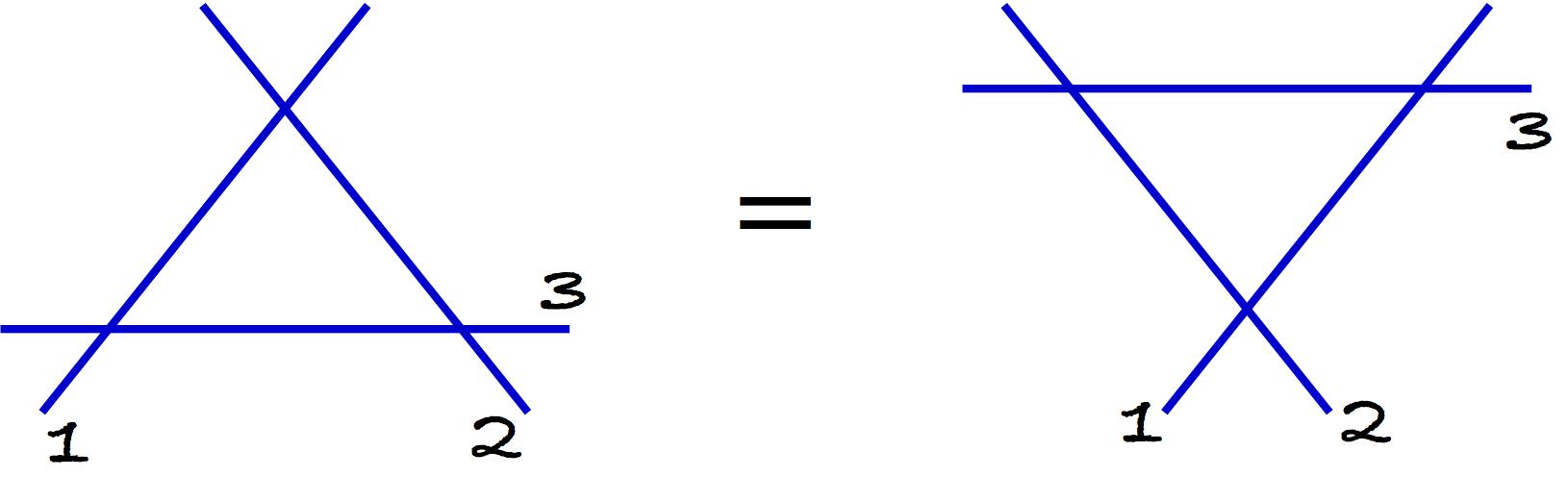}}
\caption{\label{YBE}\small The Yang-Baxter equation.}
\end{center}
\end{figure}

In spite of what the name might suggest, most of the physical strip corresponds to unphysical kinematics. 
The S-matrix is a probability amplitude of a physical scattering process only for real positive $s>4m^2$, cf.~(\ref{srap}). The rest is analytic continuation. More precisely, the scattering amplitude is defined for $s$ real positive $+i\epsilon $, according to the usual Feynman prescription. The upper side of the s-channel cut maps to the positive real semi-axis on the $\theta $-plane (fig.~\ref{planes}). 
Analytic continuation to the entire $s$-plane  leads to a set of S-matrix axioms \cite{Zamolodchikov:1978xm,Dorey:1996gd}:
\begin{itemize}
 \item Unitarity (+ real analyticity):
\begin{equation}\label{unitarity}
 S(-\theta )S(\theta )=\mathbbm{1}
\end{equation}
 \item Crossing symmetry\footnote{Physical particles in the $O(N)$ model are their own anti-particles. More generally, crossing involves charge conjugation.}:
\begin{equation}\label{crossing}
 S_{ij}^{kl}(\theta )=S_{il}^{kj}(i\pi -\theta )
\end{equation}
 \item Yang-Baxter equation (YBE), fig.~\ref{YBE}:
\begin{equation}
 S_{12}(\theta _1-\theta _2)S_{13}(\theta _1-\theta _3)S_{23}(\theta _2-\theta _3)
 =S_{23}(\theta _2-\theta _3)S_{13}(\theta _1-\theta _3)S_{12}(\theta _1-\theta _2)
\end{equation}
The YBE is a consistency condition on the factorization of the $3\rightarrow 3$ scattering into elemental $2\rightarrow 2$ processes. The outcome should not depend on the order in which the three particles collided pairwise. The labels on $S_{ab}$ indicate that its acts on the $SO(N)$ indices of particles $a$ and $b$.
\end{itemize}

\begin{exc}
Write down the YBE in the explicit index notations. 
\end{exc}

Symmetries and the YBE usually produce a skeleton of the S-matrix, a fixed matrix structure without the detailed dependence on the rapidity. Indeed, if $S_{ij}^{kl}(\theta )$ solves the YBE, then $f(\theta )S_{ij}^{kl}(\lambda \theta )$ also does so. To bootstrap the remaining ambiguity one needs to invoke crossing symmetry, unitarity and analyticity. Here we show how this program works for the $O(N)$ model.

The general form of the S-matrix consistent with the $O(N)$ symmetry is
\begin{equation}\label{index-form}
 S_{ij}^{kl}(\theta )=\sigma (\theta )
 \left(\delta _i^k\delta _j^l+\frac{1}{p(\theta )}\,\delta _i^l\delta _j^k+\frac{1}{k(\theta )}\,\delta _{ij}\delta ^{kl}\right).
\end{equation}
First, we impose the YBE.
The overall "dressing" factor $\sigma (\theta )$ then drops out. It is in principle straightforward to determine what the constraints on the two remaining functions are, albeit the calculations are rather lengthy. The algebra can be facilitated by introducing  permutation and trace operators, that act on the tensor product of two $SO(N)$ vectors as
\begin{equation}
 P\,a\otimes b=b\otimes a,\qquad K\, a\otimes b =(a\cdot b)\mathbbm{1}.
\end{equation}
In these notations,
\begin{equation}\label{operator-form}
 S=\sigma (\theta )\left(\mathbbm{1}+\frac{1}{p(\theta )}\,P+\frac{1}{k(\theta )}\,K\right).
\end{equation}
The permutation and trace operators satisfy the following braiding relations:
\begin{eqnarray}\label{addiid}
&&P_{ab}P_{bc}=P_{ac}P_{ab},
\qquad 
 K_{ab}^2=NK_{ab},
\qquad 
 K_{ab}P_{ab}=P_{ab}K_{ab}=K_{ab},
\nonumber \\
&&K_{ab}P_{cb}=P_{cb}K_{ac},
\qquad 
K_{ab}K_{cb}=P_{bc}P_{ab}K_{bc},
\qquad 
P_{ab}P_{bc}K_{ab}=P_{ac}K_{ab}.
\end{eqnarray}

 Algebraic manipulations based on (\ref{addiid}) -- the details of which we omit\footnote{When dealing with the YBE, one really appreciates the power of {\it Mathematica}.} -- reduce the YBE  to a set of three independent equations on functions $k(\theta )$ and $p(\theta )$:
\begin{eqnarray}
 &&p(u)-p(v)=p(u-v)
\nonumber \\
&&k(u)-k(v)=-p(u-v)
\nonumber \\
&&k(u)\left(k(u-v)+k(v)+p(u)+N\right)=k(v)+k(u-v)-p(u).
\end{eqnarray}
The first two equations simply tell us that $k(\theta )$ and $p(\theta )$ are linear functions:
\begin{equation}
 p(\theta )=-\lambda \theta ,
 \qquad 
 k(\theta )=\lambda \theta -\frac{1 }{2\Delta },
\end{equation}
where $\lambda $ and $\Delta $ are arbitrary constants. The last equation is then also satisfied, provided
\begin{equation}
 \Delta =\frac{1}{N-2}\,.
\end{equation}

As anticipated, the YBE fixes the functional structure of the S-matrix up to rescaling $\theta \rightarrow \lambda \theta $, and up to an overall dressing factor. The permutation and trace terms in the S-matrix map to one another under the crossing transformation, which interchanges the $SO(N)$ indices $j$ and $l$ in (\ref{index-form}). The crossing-symmetry condition (\ref{crossing}) then requires that
\begin{equation}
 p(i\pi -\theta )=k(\theta ),
\end{equation}
which fixes
\begin{equation}
 \lambda =\frac{1}{2\pi  i\Delta }\,.
\end{equation}
This determines the S-matrix up to the dressing factor:
\begin{equation}\label{finalSO(N)}
 S_{ij}^{kl}(\theta )=\sigma (\theta )
 \left(\delta _i^k\delta _j^l-\frac{2\pi i\Delta }{\theta}\,\delta _i^l\delta _j^k+\frac{ 2\pi i\Delta }{\theta -i\pi }\,\delta _{ij}\delta ^{kl}\right).
\end{equation}

The dressing factor is constrained by the unitarity and crossing conditions (\ref{unitarity}) and (\ref{crossing}). Substitution of (\ref{operator-form}) into (\ref{unitarity}), and the use of the identities $P^2=\mathbbm{1}$, $K^2=NK$, $PK=K=KP$, gives:
\begin{eqnarray}
&& S(-\theta )S(\theta )=\sigma (-\theta )\sigma (\theta )\left[
 1+\frac{1}{p(-\theta )p(\theta )}+\left(\frac{1}{p(-\theta )}+\frac{1}{p(\theta )}\right)P
 \right.
\nonumber \\
&&\left.
+\left(\frac{1}{k(-\theta )}+\frac{1}{k(\theta )}+\frac{1}{k(-\theta )p(\theta )}+\frac{1}{k(-\theta )p(\theta )}+\frac{N}{k(-\theta )k(\theta )}\right)K
 \right].
\nonumber 
\end{eqnarray}
The coefficients of $P$ and $K$ in this equation identically vanish, and the unitarity condition boils down to 
\begin{equation}
 \sigma (-\theta )\sigma (\theta )=\frac{\theta ^2}{\theta ^2+4\pi^ 2\Delta ^2}\,.
\end{equation}
The crossing symmetry simply requires that
\begin{equation}
 \sigma (i\pi -\theta )=\sigma (\theta ).
\end{equation}
We need to solve these two functional equations on $\sigma (\theta )$.

The dressing factor $\sigma (\theta )$ by itself is not an amplitude of any physical process. It is much more instructive to reformulate the crossing/unitarity equations in terms of a physical S-matrix element for  some particular scattering channel. General features of integrability bootstrap, which are important, for example, in the context of thermodynamic Bethe ansatz (TBA), become more transparent if the equations are written in terms of a physical amplitude.

To construct such we concentrate on a $U(1)$ subgroup of $SO(N)$ generated by rotations in the $(12)$ plane. An $SO(N)$ vector has two components with $U(1)$ charges $+1$ and $-1$ and $N-2$ neutral components. In terms of the fields in the path integral (\ref{O(N)pathintegral}), the charge $+1$ state corresponds to 
$$
 \phi =n_1+in_2.
$$
Because of the charge conservation and the absence of particle production the $\phi \phi \rightarrow \phi \phi $ amplitude is purely elastic, and the S-matrix element
\begin{equation}
 s(\theta )\equiv S_{\phi \phi \rightarrow \phi \phi }(\theta )
\end{equation}
is a pure phase:
\begin{equation}
 s(\theta )=\,{\rm e}\,^{i\Phi (\theta )}.
\end{equation}
From (\ref{finalSO(N)}) we find:
\begin{equation}\label{s-theta}
 s(\theta )=\sigma (\theta )\,\frac{\theta -2\pi i\Delta }{\theta }\,.
\end{equation}

The crossing/unitarity equations for $s(\theta )$ are
\begin{eqnarray}\label{c/u}
 &&s(i\pi -\theta )=s(\theta )\,\frac{\theta \left(\theta +2\pi i\Delta -i\pi \right)}{\left(\theta -2\pi i\Delta  \right)\left(\theta -i\pi \right)}
\nonumber \\
&&s(-\theta )s(\theta )=1.
\end{eqnarray}
From those we infer that
\begin{equation}
 s\left(\theta +\frac{i\pi }{2}\right)s\left(\theta -\frac{i\pi }{2}\right)
 =\frac{\left(\theta -\frac{i\pi }{2}\right)\left(\theta +\frac{i\pi }{2}-2\pi i\Delta \right)}{\left(\theta +\frac{i\pi }{2}\right)\left(\theta -\frac{i\pi }{2}+2\pi i\Delta \right)}\,.
\end{equation}
Written in terms of the phase, the equation becomes
\begin{equation}
 \Phi \left(\theta +\frac{i\pi }{2}\right)+\Phi \left(\theta -\frac{i\pi }{2}\right)
 =2\arctan\frac{2\theta }{\pi }-2\arctan\frac{2\theta }{{\pi }\left(1-4\Delta \right)}
+{\rm p.a.}
\end{equation}
The last term (to be specified shortly) reflects an ambiguity in choosing the branch of the logarithm when reformulating the problem in terms of the scattering phase, which is defined up to shifts by integer multiples of $2\pi $. The branch of the logarithm can be chosen differently for different $\theta $, and we can parameterize the phase ambiguity by writing
\begin{equation}
 {\rm p.a.}= 2\pi \sum_{k}^{}\left(
 \vartheta\left(\theta -i\alpha _k\right)+\vartheta\left(\theta +i\alpha _k\right)
 \right),
\end{equation}
where  $\vartheta(x)$ is the step function, and $\{\alpha _k\}$ is a potentially arbitrary set of parameters.

Another function, that plays an important role in all types of  Bethe ansatz equations, is the derivative of the scattering phase:
\begin{equation}
 K(\theta )=\frac{d\Phi (\theta )}{d\theta }\,.
\end{equation}
For this function we get a difference equation:
\begin{eqnarray}
 K \left(\theta +\frac{i\pi }{2}\right)+K\left(\theta -\frac{i\pi }{2}\right)
 &=&\frac{\pi }{\theta ^2+\frac{\pi ^2}{4}}
 -\frac{\pi\left(1-4\Delta \right) }{\theta ^2+\frac{\pi ^2}{4}\left(1-4\Delta \right)^2}
\nonumber \\
 &&+2\pi \sum_{k}^{}\left(\delta (\theta -i\alpha _k)+\delta (\theta +i\alpha _k)\right),
 \vphantom{\frac{\pi\left(1-4\Delta \right) }{\theta ^2+\frac{\pi ^2}{4}\left(1-4\Delta \right)^2}}
\end{eqnarray}
which can be also written as
\begin{eqnarray}
 2\cos\left(\frac{\pi }{2}\,\,\frac{d}{d\theta }\right)K(\theta )
 &=&\frac{\pi }{\theta ^2+\frac{\pi ^2}{4}}
 -\frac{\pi\left(1-4\Delta \right) }{\theta ^2+\frac{\pi ^2}{4}\left(1-4\Delta \right)^2}
\nonumber \\
 &&+2\pi \sum_{k}^{}\left(\delta (\theta -i\alpha _k)+\delta (\theta +i\alpha _k)\right).
  \vphantom{\frac{\pi\left(1-4\Delta \right) }{\theta ^2+\frac{\pi ^2}{4}\left(1-4\Delta \right)^2}}
\end{eqnarray}
 The equation becomes algebraic in the Fourier space:
\begin{equation}
 \cosh\frac{\pi \omega }{2}\,K(\omega )
 =\pi 
 \,{\rm e}\,^{-\frac{\pi }{2}\,|\omega |}
 -\pi\,{\rm e}\,^{-\frac{\pi }{2}\left(1-4\Delta \right)|\omega |}
 +2\pi \sum_{k}^{}\cosh\alpha _k\omega .
\end{equation}
Fourier transforming back to the $\theta $-representation we get:
\begin{eqnarray}
 2\pi  K(\theta )&=&
 \psi \left(1+\frac{i\theta }{2\pi }\right)
 +\psi \left(1-\frac{i\theta }{2\pi }\right)
 -\psi \left(\frac{1}{2}+\frac{i\theta }{2\pi }\right)
  -\psi \left(\frac{1}{2}-\frac{i\theta }{2\pi }\right)
\nonumber \\
 && -\psi \left(1-\Delta +\frac{i\theta }{2\pi }\right)
 -\psi \left(1-\Delta -\frac{i\theta }{2\pi }\right)
\nonumber \\
&&
 +\psi \left(\frac{1}{2}-\Delta +\frac{i\theta }{2\pi }\right)
  +\psi \left(\frac{1}{2}-\Delta -\frac{i\theta }{2\pi }\right)
  +\sum_{k}^{}\frac{4\pi\cos\alpha _k\cosh\theta }{\sinh^2\theta +\cos^2\alpha _k}\,.
\end{eqnarray}
Integrating and exponentiating the result we find:
\begin{equation}\label{rawS}
 s(\theta )=S_{\rm CDD}\left(\theta \right)\,
 \frac{
 \Gamma \left(1+\frac{i\theta }{2\pi }\right)
  \Gamma \left(\frac{1}{2}-\frac{i\theta }{2\pi }\right)
   \Gamma \left(1-\Delta -\frac{i\theta }{2\pi }\right)
    \Gamma \left(\frac{1}{2}-\Delta +\frac{i\theta }{2\pi }\right)
 }{
  \Gamma \left(1-\frac{i\theta }{2\pi }\right)
  \Gamma \left(\frac{1}{2}+\frac{i\theta }{2\pi }\right)
   \Gamma \left(1-\Delta +\frac{i\theta }{2\pi }\right)
    \Gamma \left(\frac{1}{2}-\Delta -\frac{i\theta }{2\pi }\right)
 }\,,
\end{equation}
where
\begin{equation}
 S_{\rm CDD}(\theta )=\prod_{k}^{}\frac{\sinh\theta -i\cos\alpha _k}{\sinh\theta +i\cos\alpha _k}
\end{equation}
is known as the CDD (Castillejo-Dalitz-Dyson) factor. 

The CDD factor parameterizes fundamental ambiguity in the crossing + unitarity conditions, which do not fix the S-matrix  completely. However, each CDD factor  introduces extra poles in the rapidity plane. Any pole  on the physical strip is associated with a bound state, and no new poles can be added by hand without changing the particle content of the theory. Usually this means that the minimal solution of the crossing equation is the physical one, but sometimes the minimal solution has spurious poles that have to be cancelled by CDD factors. 
The $O(N)$ model is a good example of this phenomenon.

The minimal solution without CDD factors would have a pole on the physical strip, at $\theta =i\pi (1-2\Delta )$. The pole has to be cancelled by a single CDD factor with 
$$
\alpha =\frac{\pi }{2}\left(1-4\Delta \right).
$$
After which (\ref{rawS}) becomes
\begin{equation}\label{s-final}
 s(\theta )=- \frac{
 \Gamma \left(1+\frac{i\theta }{2\pi }\right)
  \Gamma \left(\frac{1}{2}-\frac{i\theta }{2\pi }\right)
   \Gamma \left(\Delta -\frac{i\theta }{2\pi }\right)
    \Gamma \left(\frac{1}{2}+\Delta +\frac{i\theta }{2\pi }\right)
 }{
  \Gamma \left(1-\frac{i\theta }{2\pi }\right)
  \Gamma \left(\frac{1}{2}+\frac{i\theta }{2\pi }\right)
   \Gamma \left(\Delta +\frac{i\theta }{2\pi }\right)
    \Gamma \left(\frac{1}{2}+\Delta -\frac{i\theta }{2\pi }\right)
 }\,.
\end{equation}
This is the exact scattering amplitude of the $\phi $ quanta.

The dressing factor can be extracted from (\ref{s-theta}):
\begin{equation}\label{sigma-final}
 \sigma (\theta )=Q(\theta )Q(i\pi -\theta ),
 \qquad 
 Q(\theta )=\frac{
  \Gamma \left(\frac{1}{2}-\frac{i\theta }{2\pi }\right)
   \Gamma \left(\Delta -\frac{i\theta }{2\pi }\right)
 }{
  \Gamma \left(-\frac{i\theta }{2\pi }\right)
     \Gamma \left(\frac{1}{2}+\Delta -\frac{i\theta }{2\pi }\right)
 }\,.
\end{equation}
Crossing symmetry is manifest in (\ref{sigma-final}), while unitarity is manifest in (\ref{s-final}). 

\begin{exc}
 Expand the exact S-matrix of the $O(N)$ model to the first non-vanishing order in $1/N$, and compare with the result of the explicit calculation from exercise \ref{O(N)S-matrix:explicit}.
\end{exc}

\section{Crash course in superalgebras}\label{CCIS:sec}

String theory requires supersymmetry for internal consistency. If a sigma-model is to describe critical ten-dimensional string it has to be supersymmetric. Supersymmetric counterparts of the coset sigma-models that we studied so far are ubiquitous in holographic duality, and many of them are integrable, just like in purely bosonic case. Fields in these models take values in a supermanifold, and are invariant under the action of a supergroup. Any detailed introduction to supergroups and superalgebras (which can be found in \cite{Kac:1977qb,Frappat:1996pb}) would occupy too large space-time volume, but the most common examples can be introduced very easily,  as algebras of matrices of certain type, just like  most common examples of ordinary Lie algebras are algebras of matrices constrained to be traceless, anti-symmetric, and so on.

In addition to bosonic generators $T_M$, a superalgebra contains supercharges $Q_\alpha $. An element of a superalgebra is a linear combination
$$
 X=X^MT_M+\theta ^\alpha Q_\alpha ,
$$
where $\theta ^\alpha $ are Grassmann-odd variables. The commutator on a superalgebra must satisfy the usual requirements of linearity, anti-symmetry and obey the Jacobi identity. A group element is defined by an exponential map $g=\,{\rm e}\,^{X}$, which is actually polynomial in supercharges because a sufficiently large power of Grassmann variables vanishes identically.

In any finite-dimensional representation  the generators $T_M$ and $Q_\alpha $ are  ordinary numeric matrices. An element of the Lie superalgebra  $X$ then forms a supermatrix -- a matrix some of whose elements are Grassmann-odd numbers. 

Complex supermatrices can be defined as linear transformations of $\mathbbm{C}^{(n|m)}$, a space of vectors whose first $n$ components are Grassmann-even and the last $m$ ones are Grassmann-odd:
$$
 \left(\,
\begin{array}{ c : c }
 B & \Theta  \\\hdashline
 \Theta ' & B' \\  
\end{array}\,\right)
\left(\,
\begin{array}{ c }
 b \\\hdashline
 f \\  
\end{array}\,\right).
$$
To preserve the Grassmann-parity structure, the diagonal blocks of a supermatrix must be even and the off-diagonal blocks must be odd.

Supermatrices can be multiplied and commuted, with the commutator obeying all the usual axioms of the Lie bracket. The algebra of all $(n|m)\times (n|m)$ supermatrices is denoted by $\mathfrak{gl}(n|m)$. This superalgebra is not simple. An obvious way to get a simple algebra is to impose the traceless condition, and here the first deviation from the usual Lie algebras occurs. 

Consider this simple example:
$$
 \left[\begin{pmatrix}
  0 & \theta _1 \\ 
  \psi _1 & 0 \\ 
 \end{pmatrix},
 \begin{pmatrix}
  0 & \theta _2 \\ 
  \psi _2 & 0 \\ 
 \end{pmatrix}\right]
 =\begin{pmatrix}
  \theta _1\psi _2-\theta _2 \psi _1 & 0 \\ 
  0  & \psi _1 \theta _2- \psi _2\theta _1 \\ 
 \end{pmatrix}.
$$
When $\theta _i$, $\psi _i$ are ordinary numbers the sum of diagonal matrix elements on the right-hand side vanishes: the trace of a commutator is always zero. But if $\theta$'s and $\psi$'s are Grassmann numbers, the diagonal matrix elements differ by a sign and the trace does not vanish. Since the trace of a commutator of supermatrices is not zero, we conclude that traceless supermatrices do not form a subalgebra in $\mathfrak{gl}(1|1)$. 

This can be fixed in the following way. In the above example the difference, rather then the sum, of the diagonal matrix elements is zero. This is know as supertrace, in general defined as
\begin{equation}
 \mathop{\mathrm{Str}}
  \left(\,
\begin{array}{ c : c }
 B & \Theta  \\\hdashline
 \Theta ' & B' \\  
\end{array}\,\right)=\mathop{\mathrm{tr}}B-\mathop{\mathrm{tr}}B'.
\end{equation}

\begin{exc}
Show that  for any supermatrices
\begin{equation}
 \mathop{\mathrm{Str}}[M_1,M_2]=0.
\end{equation}
\end{exc}

Because of the last property, restriction to {\it supertraceless} supermatrices gives a closed algebra, known as the special linear superalgebra $\mathfrak{sl} (n|m)$. The $\mathfrak{sl}(n|m)$ superalgebras are simple if $m\neq n$, but the  case of $m=n$ is special. Since 
$$
 \mathop{\mathrm{Str}}
  \left(\,
\begin{array}{ c : c }
 \mathbbm{1} & 0  \\\hdashline
 0 & \mathbbm{1} \\  
\end{array}\,\right)=n-m,
$$
for $m=n$ the trace condition does not eliminate the unit matrix! The unit matrix commutes with anything and thus represents a central element of the algebra. The factor-algebra, where this central element is set to zero is denoted by
\begin{equation}
 \mathfrak{psl}(n|n)=\mathfrak{sl}(n|n)/C.
\end{equation}
In any irreducible representation the central element takes a constant numeric value, and $\mathfrak{psl(n|n)}$ is defined by keeping only those representations of  $\mathfrak{sl}(n|n)$ in which the central element is equal to zero. It is interesting to notice that the defining $(n|n)$-dimensional representation of $\mathfrak{sl}(n|n)$ is not a representation of $\mathfrak{psl}(n|n)$.

The $\mathfrak{psl}(n|n)$  superalgebras are remarkable in many respects. Being finite-dimensional they admit a central extension (never happens for ordinary Lie algebras). They also have zero dual Coxeter number, where the dual Coxeter number is defined as in the bosonic case through eq.~(\ref{dualCoxeter}). To see this consider a diagonal supermatrix $K=\mathop{\mathrm{diag}}(a_1\ldots a_n|b_1\ldots b_n)$, which represents an element of $\mathfrak{psl}(n|n)$ if the sums of $a_i$'s and $b_I$'s are separately zero. As for it's action in the adjoint representation, the eigenvalues of  $\mathop{\mathrm{ad}}K$ are $a_i-a_j$, $b_I-b_J$, $a_i-b_J$ and $b_I-a_j$ for all unordered pairs of indices $ij$, $IJ$, $iJ$ and $Ij$. Therefore,
$$
 \mathop{\mathrm{Str}}\nolimits_{\rm adj}K^2=\sum_{ij}^{}(a_i-a_j)^2+
 \sum_{IJ}^{}(b_I-b_J)^2-\sum_{iJ}^{}(a_i-b_J)^2-\sum_{Ij}^{}(b_I-a_j)^2
 =0.
$$

The unitary superalgebras $\mathfrak{(p)su}(p,q|r,s)$  are real forms of $\mathfrak{(p)sl}(n|m)$ singled out by the Hermiticity condition
\begin{equation}
 M^\dagger =-B^{-1}MB,
 \qquad 
 B=\mathop{\mathrm{diag}}\left(+^p-^q|+^r-^s\right).
\end{equation}

Another important class of superalgebras is the supersymmetric counterpart of $\mathfrak{so}(n)$. Ordinary orthogonal algebra is singled out by the anti-symmetry condition imposed on elements of $\mathfrak{gl}(n)$. For supermatrices this does not work because transposition is not compatible with commutators: in general, $(M_1M_2)^t\neq M_2^tM_1^t$ for supermatrices, since their off-diagonal Grassmann entries do not commute but anti-commute. 

The supersymmetric generalization of transposition, called  supertransposition, is defined as
\begin{equation}\label{supertransposition}
 \begin{pmatrix}
  A & \Theta  \\ 
  \Psi  & B \\ 
 \end{pmatrix}^{st}
 =\begin{pmatrix}
 A^t  & -\Psi ^t   \\ 
  \Theta ^t &  B^t \\ 
 \end{pmatrix}.
\end{equation}
It follows that
\begin{equation}
 (M_1M_2)^{st}=M_2^{st}M_1^{st},
\end{equation}
and consequently supertransposition preserves commutators up to a sign.

An orthosymplectic algebra $\mathfrak{osp}(n|2m)$ is defined as a subset of $\mathfrak{gl}(n|2m)$ singled out by the condition
\begin{equation}
 M^{st}=-H^{-1}MH,
\end{equation}
where
\begin{equation}
 H=\begin{pmatrix}
  \mathbbm{1}_{n\times n} & 0 \\ 
  0 & J \\ 
 \end{pmatrix},
\end{equation}
and $J$ is the $2m\times 2m$ symplectic matrix:
\begin{equation}\label{symplecticM}
 J=\begin{pmatrix}
  0 & \mathbbm{1}_{m\times m} \\ 
   -\mathbbm{1}_{m\times m} & 0 \\ 
 \end{pmatrix}.
\end{equation}
Explicitly,
\begin{equation}
 \mathfrak{osp}(n|2m)=\left\{
 \begin{pmatrix}
  A & \Theta  \\ 
  \Psi  & B \\ 
 \end{pmatrix}\in \mathfrak{gl}(n|2m)\right|\left.
 A^t=-A,\,B^t=JBJ,\,\Theta ^t =J\Psi  ^t
 \vphantom{ \begin{pmatrix}
  A & \Theta  \\ 
  \Psi  & B \\ 
 \end{pmatrix}}
 \right\}.
\end{equation}
The bosonic subalgebra of $\mathfrak{osp}(n|2m)$  is $\mathfrak{so}(n)\oplus\mathfrak{sp}(2m)$.

\begin{exc}
Check that superalgebras $\mathfrak{osp}(2n+2|2n)$ have zero dual Coxeter number.
\end{exc}

\section{Supercoset sigma-models}\label{supercosets:sec}

A natural supersymmetric generalization of a symmetric space is based on an automorphism $\Omega :\mathfrak{g}\rightarrow \mathfrak{g}$ that squares to the fermion parity:
\begin{equation}
 \Omega ^2=\left(-1\right)^F.
\end{equation}
Such automorphism has order four: $\Omega ^4={\rm id}$, and the superalgebra has a $\mathbbm{Z}_4$ grading:
\begin{equation}
 \mathfrak{g}=\mathfrak{h}_0\oplus\mathfrak{h}_1\oplus\mathfrak{h}_2\oplus\mathfrak{h}_3,
\end{equation}
where
\begin{equation}
 \Omega (\mathfrak{h}_n)=i^4\mathfrak{h}_n.
\end{equation}
The $\mathbbm{Z}_4$ decomposition is consistent with the Grassmann parity: $\mathfrak{h}_0\oplus\mathfrak{h}_2$ form the bosonic subalgebra of $\mathfrak{g}$ and all the supercharges belong to $\mathfrak{h}_1\oplus\mathfrak{h}_3$.

If $H_0$ is the subgroup of $G$ whose  Lie algebra $\mathfrak{h}_0$ is $\mathbbm{Z}_4$-invariant, the coset $G/H_0$ is called semi-symmetric superspace \cite{Serganova}. Anti-de-Sitter geometries in various dimensions arise as bosonic sections of $\mathbbm{Z}_4$ cosets, which play an important role in the holographic duality. The key example is string theory on $AdS_5\times S^5$ whose sigma-model \cite{Metsaev:1998it} is a semi-symmetric coset \cite{Berkovits:1999zq}.
Just like for ordinary symmetric-space sigma-models, integrability of semi-symmetric cosets can be established at the algebraic level, using constraints imposed on the equations of motion by the $\mathbbm{Z}_4$-symmetry.

The current (\ref{LIcurrent}) decomposes into four components according to their $\mathbbm{Z}_4$ grading:
\begin{equation}
 j=j_0+j_1+j_2+j_3.
\end{equation}
The current components $j_0$, $j_2$ expand in bosonic generators of the superalgebra, while the components $j_1$, $j_3$ are fermionic -- they are linear combinations of supercharges with Grassmann-odd coefficients.

The action  of the $\mathbbm{Z}_4$ coset sigma-model is
\begin{equation}\label{z4cosetaction}
 S=\int_{}^{}d^2x\,\mathop{\mathrm{Str}}\left(
 \sqrt{-h}h^{ab}j_{2\,a}j_{2\,b}+\varepsilon ^{ab}j_{1\,a}j_{3\,b}
 \right).
\end{equation}
The first term is the usual sigma-model Lagrangian, while the second term descends from the WZ action, which for supercosets is not topological and can be  written in a manifestly 2d form.

\begin{exc}
Derive the equations of motion for a $\mathbbm{Z}_4$ coset and check that the combined system of equations of motion and the flatness condition of the left-invariant current admits  zero-curvature representation with the Lax connection \cite{Bena:2003wd}
\begin{equation}
 L=j_0+\frac{z^2+1}{z^2-1}\,j_2-\frac{2z}{z^2-1}\,*j_2+\sqrt{\frac{z+1}{z-1}}\,j_1+\sqrt{\frac{z-1}{z+1}}\,j_3.
\end{equation}
Therefore the model is completely integrable.
\end{exc}

The action (\ref{z4cosetaction}), taken at face value, is quadratic in derivatives, which is fine for bosonic degrees of freedom, but a fermion Lagrangian is normally of Dirac type and is expected to be linear in derivatives. It is instructive to see how this apparent contradiction is resolved in $\mathbbm{Z}_4$ cosets.

To get some intuition about the structure of the sigma-model action we can take the coset representative in the standard exponential form: 
\begin{equation}\label{expcoset}
 g=\,{\rm e}\,^{X},\qquad X=X_1+X_2+X_3=\theta ^\alpha _1Q_{_1\,\alpha }+X^MT_M+\theta ^{\dot{\alpha }} _3Q_{_3\,\dot{\alpha } },
\end{equation}
and expand the currents in $X$ to the linear order. Then $j_{n\,a}\simeq \partial_a X_n$. The sigma-model term becomes $(\partial X^M)^2$, while the WZ term gives $\varepsilon ^{ab}\partial _aX_1\partial _bX_3$, which looks like a second-order kinetic term for fermions, but in fact is a total derivative that integrates to zero:
$$
 \varepsilon ^{ab}\partial _aX_1\partial _bX_3=\partial _a\left(\varepsilon ^{ab}X_1\partial _bX_3\right).
$$
Quadratic term for fermions appears only at the next order in the expansion.

The cubic part of the action, one can show, contains\footnote{To simplify the notations we consider flat Minkowski metric on the worldsheet from now on.}
\begin{eqnarray}\label{GSaction}
 \mathcal{L}_F^{(3)}&=&\mathop{\mathrm{Str}}\left(
 \Pi _-^{ab}X_1[\partial _aX_2,\partial _bX_1]
 +\Pi _+^{ab}X_3[\partial _aX_2,\partial _bX_3]
 \right)
\nonumber \\
 &=&\Pi _-^{ab}\theta _{1\,\alpha} f_{M\beta }^\alpha \partial_a X^M\partial_b\theta _1^\beta
 +\Pi _+^{ab}\theta _{3\,\dot{\alpha }} f_{M\dot{\beta }}^{\dot{\alpha }} \partial_a X^M\partial_b\theta _3^{\dot{\beta }},  
\end{eqnarray}
where $f_{AB}^C$ are structure constants of the superalgebra $\mathfrak{g}$, and $\Pi _\pm$ are the worldsheet chirality projectors:
\begin{equation}
 \Pi _\pm^{ab}=\eta ^{ab}\pm \varepsilon ^{ab}.
\end{equation}
This action is exactly the same as the quadratic part of the Green-Schwarz (GS) action for the superstring \cite{Green:1983wt}, with the structure constants $f_{M\alpha }^\beta $ replacing the Dirac matrices.

The GS action is notoriously difficult to quantize covariantly, one manifestation of which is absence of kinetic terms for fermions. Those can be induced if bosonic currents have a non-zero expectation value: $\left\langle \partial X^M\right\rangle\neq 0$. The most familiar instance when this happens is the light-cone gauge:
\begin{equation}\label{LCgauge}
 X^+=\tau .
\end{equation}
Taking into account that
\begin{equation}
 \Pi _\pm^{0b}\partial _b=\partial _\pm,
\end{equation}
where $\partial _\pm$ are the worldsheet light-cone derivatives, we see that the GS action in the light-cone gauge describes the standard left- and right-moving Majorana fermions:
\begin{equation}\label{lcfermions}
 \mathcal{L}_{F,{\rm lc}}^{(3)}=\theta _1f_+\partial _-\theta _1+\theta _3f_+\partial _+\theta _3.
\end{equation}

\begin{exc}\label{bckrZ4}
 Consider a slightly more general setup. As in sec.~\ref{QSM:sec}, take the coset representative in the form (\ref{backgroundexp}) where $\bar{g}$ is the background field that belongs to the bosonic subgroup of $G$ and satisfies the equations of motion, and expand the action to the quadratic order in fluctuations.  Reproduce the trilinear terms in the action (\ref{GSaction}) from this more general result. Show that  the light-cone gauge-fixed action in addition to (\ref{lcfermions}) contains a mass term that couples $\theta _1$ and $\theta _3$. 
\end{exc}

\begin{exc}
Use the results of the exercise \ref{bckrZ4} to compute the one-loop beta function for generic semi-symmetric coset \cite{Zarembo:2010sg}. Show that the beta function vanishes as soon as the superalgebra $\mathfrak{g}$ has zero dual Coxeter number ($\mathop{\mathrm{Str}}\nolimits_{{\rm adj}}K^2=0$ $\forall~K\in\mathfrak{h}_2$). 
 \end{exc}
 
Semi-symmetric superspaces are all classified \cite{Serganova}. As we have seen each of them gives rise to an integrable GS-type sigma-model. Not  all of these sigma-models are consistent as string backgrounds. There is a number of additional constraints that need to be satisfied. The sigma-model must be conformally invariant and must have zero beta function, to begin with. The one-loop beta function of semi-symmetric cosets is proportional to the supertrace of the background current in the adjoint, and only those cosets whose dual Coxeter number vanishes have a chance to be consistent as string sigma-models.  Two infinite series of superalgebras with vanishing adjoint supertrace were identified in sec.~\ref{CCIS:sec}: $\mathfrak{psl}(n|n)$ and $\mathfrak{osp}(2n+2|2n)$. 

To define critical superstring in ten dimensions a sigma-model should also have central charge $c=26$, which is more or less equivalent to the condition that the physical degrees of freedom in the light-cone gauge are $8_b+8_f$ transverse fluctuation modes of the string. We are not going to present a comprehensive analysis of all possible case, but restrict exposition to two representative examples.

The $\mathfrak{psl}(m|m)$ superalgebras admit a number of $\mathbbm{Z}_4$ automorphisms. One of them, defined for even $m=2n$, is a combination of supertransposition (\ref{supertransposition}) and a conjugation:
\begin{equation}
 \Omega =-st\circ \mathop{\mathrm{Ad}}\mathop{\mathrm{diag}}(J,J),
\end{equation}
where $J$ is the symplectic matrix (\ref{symplecticM}). The invariant subalgebra of $\Omega $ is
\begin{equation}
 \mathfrak{h}_0=\left\{\mathop{\mathrm{diag}}(A,B)\right|\left.
 A=JA^tJ,\,B=JB^tJ
 \right\},
\end{equation}
namely,
\begin{equation}
 \mathfrak{h}_0=\mathfrak{sp}(2n)\oplus\mathfrak{sp}(2n).
\end{equation}
The bosonic dimension of the coset (the number of generators in $\mathfrak{h}_2$) is $2[(4n^2-1)-(2n^2+n)]=4n^2-2n-2$. For $n=2$, this equals $10$.

Choosing the real form to be $\mathfrak{psu}(2,2|4)$, and taking into account that $SU(4)=Spin(6)$, $SU(2,2)=Spin(4,2)$, $USp(4)=Spin(5)$ and $USp(2,2)=Spin(4,1)$ we get the coset $PSU(2,2|4)/Spin(4,1)\times Spin(5)$ whose bosonic section is
$$
 Spin(4,2)/Spin(4,1)\times Spin(6)/Spin(5)= AdS_5\times S^5.
$$
The coset therefore supersymmetrizes  $AdS_5\times S^5$, the background that plays so prominent r\^ole in the AdS/CFT correspondence.

The off-diagonal blocks of a $\mathfrak{psu}(2,2|4)$ element supply $2\times 4^2=32$ real supercharges.  In the coset parameterization (\ref{expcoset}) those give rise to $32$ fermion fields in the sigma-model. However, only half of them correspond to physical, propagating degrees of  freedom. This is readily seen in the light-cone gauge associated with $T_+=(++--|++--)\sim D-J$, where $D$ is the dilatation generator of the conformal algebra and $J$ is the $\mathfrak{so}(6)$ angular momentum. So chosen $T_+$ is indeed a null element of $\mathfrak{psu}(2,2|4)$:
$\mathop{\mathrm{Str}}T_+^2=0$.
The light-cone coordinate $X^+$ is a combination of global AdS time and an angle on $S^5$, so the classical solution (\ref{LCgauge}) describes  a point-like string moving along the big circle of $S^5$ at the speed of light \cite{Berenstein:2002jq}.
The kinetic terms in the gauge-fixed fermion action contain $f_+=\mathop{\mathrm{ad}}T_+$. When restricted to the fermionic subspace of $\mathfrak{psu}(2,2|4)$, $\mathop{\mathrm{ad}}T_+$ has eight eigenvalues $+2$, eight eigenvalues $-2$ and sixteen zero eigenvalues. Half of the fermions (those corresponding to the zero modes of $\mathop{\mathrm{ad}}T_+$) do not appear in the action at all. This is a manifestation of the kappa-symmetry, a Grassmann-odd gauge symmetry that at the linearized level acts as a shift symmetry: $\theta \rightarrow \theta_\alpha  +\varepsilon \kappa_\alpha  $, where $\kappa _\alpha $ are zero eigenvectors of $\mathop{\mathrm{ad}}T_+$. Gauge-fixing the kappa-symmerty leaves 16 Majorana fermions, 8 of which are left- and 8 are right-moving. This is the spectrum of  critical superstring in ten dimensions.

Another example is the sigma-model based on the $\mathfrak{osp}(2n+2|2n)$ series of superalgebras. These superalgebras admit a $\mathbbm{Z}_4$ automorphism
\begin{equation}
 \Omega =\mathop{\mathrm{Ad}}\mathop{\mathrm{diag}}(J|+^p-^{n-p}+^p-^{n-p}).
\end{equation}
The invariant subalgebra is $\mathfrak{h}_0=\mathfrak{u}(n+1)\oplus\mathfrak{sp}(2p)\oplus\mathfrak{sp}(2n-2p)$, and the bosonic dimension of the coset is $(n+1)(2n+1)-(n+1)^2+(2n^2+n)-(2p^2+p)-[2(n-p)^2+n-p]=n^2+n+4np-4p^2$, which gives ten for $n=2$ and $p=1$ (for larger $n$ the dimension is too big). Taking into account that the appropriate real form of $\mathfrak{sp}(4)$ is $\mathfrak{so}(3,2)$, that $\mathfrak{so}(6)=\mathfrak{su}(3)$, and that the denominator subalgebra is $\mathfrak{u}(3)\oplus\mathfrak{so}(3,1)$ we get the coset $OSp(6|4)/U(3)\times SO(3,1)$ whose bosonic section is
\begin{equation}
 SU(4)/U(3)\times SO(3,2)/SO(3,1)=\mathbbm{CP}^3\times AdS_4,
\end{equation}
another representative holographic background.

\begin{exc}
Show that the number of kappa symmetries in the \\ $OSp(6|4)/U(3)\times SO(3,1)$ coset is such that the light-cone spectrum contains $8_L+8_R$ Majorana fermions.
\end{exc}

\subsection*{Acknowledgements}
I would like to thank the organizers of the Les Houches summer school "Integrability: From Statistical Systems To Gauge Theory"  for the opportunity to present these lectures. I am grateful to X.~Chen-Lin, D.~Medina-Rincon and E.~Wid\'en for carefully reading the notes and for many suggestions for improvements.
This work was supported by the Marie
Curie network GATIS of the European Union's FP7 Programme under REA Grant
Agreement No 317089, by the ERC advanced grant No 341222, by the Swedish Research Council (VR) grant
2013-4329, by the grant "Exact Results in Gauge and String Theories" from the Knut and Alice Wallenberg foundation, and by RFBR grant 15-01-99504. 

\appendix

\bibliographystyle{nb}

\begin{thebibliography}{10}
\ifx\href\asklfhas\newcommand{\href}[2]{#2}\fi
\raggedright
\small
\parskip 0pt

\bibitem{Weinberg:1966fm}
S.~Weinberg,
\textit{``{Dynamical approach to current algebra}''},
\textsf{Phys.~Rev.~Lett.~18,~188~(1967)}.
%
\bibitem{Weinberg:1968de}
S.~Weinberg,
\textit{``{Nonlinear realizations of chiral symmetry}''},
\textsf{Phys.~Rev.~166,~1568~(1968)}.
%
\bibitem{Haldane:1982rj}
F.~D.~M.~Haldane,
\textit{``{Continuum dynamics of the 1-D Heisenberg antiferromagnetic
  identification with the O(3) nonlinear sigma model}''},
\textsf{Phys.~Lett.~A93,~464~(1983)}.
%
\bibitem{Haldane:1983ru}
F.~D.~M.~Haldane,
\textit{``{Nonlinear field theory of large spin Heisenberg antiferromagnets.
  Semiclassically quantized solitons of the one-dimensional easy Axis Neel
  state}''},
\textsf{Phys.~Rev.~Lett.~50,~1153~(1983)}.
%
\bibitem{Tsvelik-book}
A.~M.~Tsvelik,
\textit{``Quantum field theory in condensed matter physics''},
Cambridge Univ. Press (2003).
%
\bibitem{Zamolodchikov:1978xm}
A.~B.~Zamolodchikov and A.~B.~Zamolodchikov,
\textit{``{Factorized S-matrices in two dimensions as the exact solutions of
  certain relativistic quantum field models}''},
\textsf{Annals~Phys.~120,~253~(1979)}.
%
\bibitem{Metsaev:1998it}
R.~R.~Metsaev and A.~A.~Tseytlin,
\textit{``Type IIB superstring action in $AdS_5 \times S^5$ background''},
\textsf{Nucl.~Phys.~B533,~109~(1998)},
\href{http://arXiv.org/abs/hep-th/9805028}{\texttt{hep-th/9805028}}.
%
\bibitem{Bena:2003wd}
I.~Bena, J.~Polchinski and R.~Roiban,
\textit{``Hidden symmetries of the {$AdS_5\times S^5$} superstring''},
\textsf{Phys.~Rev.~D69,~046002~(2004)},
\href{http://arXiv.org/abs/hep-th/0305116}{\texttt{hep-th/0305116}}.
%
\bibitem{Dubrovin-book}
B.~A.~Dubrovin, A.~T.~Fomenko and S.~P.~Novikov,
\textit{``Modern geometry -- methods and applications: Part ii: The geometry
  and topology of manifolds''},
Springer (1985).
%
\bibitem{Faddeev:1987ph}
L.~D.~Faddeev and L.~A.~Takhtajan,
\textit{``Hamiltonian methods in the theory of solitons''},
Springer (1987).
%
\bibitem{Zakharov:1973pp}
V.~E.~Zakharov and A.~V.~Mikhailov,
\textit{``{Relativistically Invariant Two-Dimensional Models in Field Theory
  Integrable by the Inverse Problem Technique}''},
\textsf{Sov.~Phys.~JETP~47,~1017~(1978)}.
%
\bibitem{Faddeev:1985qu}
L.~D.~Faddeev and N.~Y.~Reshetikhin,
\textit{``{Integrability of the principal chiral field model in
  (1+1)-dimension}''},
\textsf{Ann.~Phys.~167,~227~(1986)}.
%
\bibitem{Faddeev:1996iy}
L.~D.~Faddeev,
\textit{``{How Algebraic Bethe Ansatz works for integrable model}''},
\href{http://arXiv.org/abs/hep-th/9605187}{\texttt{hep-th/9605187}}.
%
\bibitem{Callan:1969sn}
C.~G.~Callan,~Jr., S.~R.~Coleman, J.~Wess and B.~Zumino,
\textit{``{Structure of phenomenological Lagrangians. 2.}''},
\textsf{Phys.~Rev.~177,~2247~(1969)}.
%
\bibitem{Eichenherr:1979ci}
H.~Eichenherr and M.~Forger,
\textit{``{On the Dual Symmetry of the Nonlinear Sigma Models}''},
\textsf{Nucl.~Phys.~B155,~381~(1979)}.
%
\bibitem{Bykov:2016rdv}
D.~Bykov,
\textit{``{Complex structures and zero-curvature equations for
  $\sigma$-models}''},
\textsf{Phys.~Lett.~B760,~341~(2016)},
\href{http://arXiv.org/abs/1605.01093}{\texttt{1605.01093}}.
%
\bibitem{Novikov:1982ei}
S.~Novikov,
\textit{``{The Hamiltonian formalism and a many valued analog of Morse
  theory}''},
\textsf{Usp.Mat.Nauk~37N5,~3~(1982)}.
%
\bibitem{Witten:1983tw}
E.~Witten,
\textit{``{Global Aspects of Current Algebra}''},
\textsf{Nucl.Phys.~B223,~422~(1983)}.
%
\bibitem{Abanov:2017zok}
A.~G.~Abanov,
\textit{``{Topology, geometry and quantum interference in condensed matter
  physics}''},
\href{http://arXiv.org/abs/1708.07192}{\texttt{1708.07192}}.
%
\bibitem{Polyakov:1975yp}
A.~M.~Polyakov and A.~A.~Belavin,
\textit{``{Metastable States of Two-Dimensional Isotropic Ferromagnets}''},
\textsf{JETP~Lett.~22,~245~(1975)}.
%
\bibitem{Witten:1983ar}
E.~Witten,
\textit{``{Nonabelian Bosonization in Two-Dimensions}''},
\textsf{Commun.Math.Phys.~92,~455~(1984)}.
%
\bibitem{TakhtajanVeselov84}
A.~P.~Veselov and L.~A.~Takhtajan,
\textit{``{Integrability of the Novikov equations for principal chiral fields
  with a multivalued Lagrangian}''},
\textsf{Sov.~Phys.~Dokl.~29,~994~(1984)}.
%
\bibitem{Polyakov:1975rr}
A.~M.~Polyakov,
\textit{``{Interaction of Goldstone Particles in Two-Dimensions. Applications
  to Ferromagnets and Massive Yang-Mills Fields}''},
\textsf{Phys.Lett.~B59,~79~(1975)}.
%
\bibitem{Friedan:1980jf}
D.~Friedan,
\textit{``{Nonlinear Models in Two Epsilon Dimensions}''},
\textsf{Phys.~Rev.~Lett.~45,~1057~(1980)}.
%
\bibitem{Zarembo:2008hb}
K.~Zarembo,
\textit{``{Quantum Giant Magnons}''},
\textsf{JHEP~0805,~047~(2008)},
\href{http://arXiv.org/abs/0802.3681}{\texttt{0802.3681}}.
%
\bibitem{Dorey:1996gd}
P.~Dorey,
\textit{``{Exact S matrices}''},
\href{http://arXiv.org/abs/hep-th/9810026}{\texttt{hep-th/9810026}}.
%
\bibitem{Kac:1977qb}
V.~G.~Kac,
\textit{``{A Sketch of Lie Superalgebra Theory}''},
\textsf{Commun.~Math.~Phys.~53,~31~(1977)}.
%
\bibitem{Frappat:1996pb}
L.~Frappat, P.~Sorba and A.~Sciarrino,
\textit{``{Dictionary on Lie superalgebras}''},
\href{http://arXiv.org/abs/hep-th/9607161}{\texttt{hep-th/9607161}}.
%
\bibitem{Serganova}
V.~V.~Serganova,
\textit{``{Classification of real simple Lie superalgebras and symmetric
  superspaces}''},
\textsf{Funct.~Anal.~Appl.~17,~200~(1983)}.
%
\bibitem{Berkovits:1999zq}
N.~Berkovits, M.~Bershadsky, T.~Hauer, S.~Zhukov and B.~Zwiebach,
\textit{``{Superstring theory on $AdS_2\times S^2$ as a coset
  supermanifold}''},
\textsf{Nucl.~Phys.~B567,~61~(2000)},
\href{http://arXiv.org/abs/hep-th/9907200}{\texttt{hep-th/9907200}}.
%
\bibitem{Green:1983wt}
M.~B.~Green and J.~H.~Schwarz,
\textit{``{Covariant Description of Superstrings}''},
\textsf{Phys.~Lett.~B136,~367~(1984)}.
%
\bibitem{Zarembo:2010sg}
K.~Zarembo,
\textit{``{Strings on Semisymmetric Superspaces}''},
\textsf{JHEP~1005,~002~(2010)},
\href{http://arXiv.org/abs/1003.0465}{\texttt{1003.0465}}.
%
\bibitem{Berenstein:2002jq}
D.~E.~Berenstein, J.~M.~Maldacena and H.~S.~Nastase,
\textit{``{Strings in flat space and pp waves from N = 4 super Yang Mills}''},
\textsf{JHEP~0204,~013~(2002)},
\href{http://arXiv.org/abs/hep-th/0202021}{\texttt{hep-th/0202021}}.
%
\end{thebibliography}

\end{document}